\documentclass{article}
\usepackage{kr}

\newif\iffull

\fullfalse

\usepackage[normalem]{ulem}  %

\usepackage{times}
\usepackage{soul}
\usepackage{url}
\usepackage[hidelinks]{hyperref}
\usepackage[utf8]{inputenc}
\usepackage[small]{caption}
\usepackage{graphicx}
\usepackage{amsmath}
\usepackage{amsthm}
\usepackage{booktabs}
\usepackage{algorithmic}
\usepackage{graphicx} %
\usepackage{thmtools}
\usepackage{microtype}

\usepackage{mathtools}
\usepackage{todonotes}
\usepackage{amsfonts}
\usepackage{amssymb}
\usepackage{enumitem}
\usepackage[color]{changebar}
\usepackage{stmaryrd} %
\usepackage{xspace}
\usepackage{cleveref}
\usepackage{multirow,bigdelim}
\usepackage{comment}
\usepackage{arydshln} %
\usepackage{complexity}
\usepackage{dblfloatfix}
\usepackage{tikz}
\usetikzlibrary{positioning}
\usepackage[linesnumbered, lined]{algorithm2e}
\usepackage{tcolorbox}
\usepackage{multicol}

\newtheorem{theorem}{Theorem}
\newtheorem{corollary}[theorem]{Corollary}
\newtheorem{lemma}[theorem]{Lemma}
\newtheorem{proposition}[theorem]{Proposition}

\newtheorem{remark}[theorem]{Remark}

\newtheorem{manualtheoreminner}{Theorem}
\newenvironment{manualtheorem}[1]{%
  \IfBlankTF{#1}
    {\renewcommand{\themanualtheoreminner}{\unskip}}
    {\renewcommand\themanualtheoreminner{#1}}%
  \manualtheoreminner
}{\endmanualtheoreminner}

\newtheorem{manualpropinner}{Proposition}
\newenvironment{manualprop}[1]{%
  \IfBlankTF{#1}
    {\renewcommand{\themanualpropinner}{\unskip}}
    {\renewcommand\themanualpropinner{#1}}%
  \manualpropinner
}{\endmanualpropinner}

\theoremstyle{definition}
 \newtheorem{definition}[theorem]{Definition}
 \newtheorem{example}[theorem]{Example}

\newcommand{\N}{\mathbb{N}}

\newcommand{\opgraph}[1]{\mathrm{OPGraph}(#1)}
\newcommand{\set}[1]{\bigcup{#1}}

\newcommand{\ledge}[1]{\xrightarrow{#1}}

\newcommand{\Var}{\textrm{Var}}

\newcommand{\Cont}{\textrm{Conts}}

\newcommand{\as}{\textrm{As}}
\newcommand{\dom}{\mathrm{Dom}}
\newcommand{\Dom}{\mathrm{Dom}}
\newcommand{\restr}{\!\restriction\!}

\newcommand{\pX}{\mathcal{R}}
\newcommand{\pY}{\mathcal{S}}

\newcommand{\pZ}{\mathcal{T}}

\newcommand{\fd}[2]{\ensuremath{#1 \rightarrow #2}} %

\newcommand{\KX}{R} 
\newcommand{\KY}{S} 
\newcommand{\KZ}{T} 
\newcommand{\Vars}{\mathrm{Vars}}

\newcommand{\supp}{\mathrm{Supp}}

\newcommand{\ccr}{\ensuremath{\texttt{ccr}}}

\newcommand{\Rp}{\ensuremath{\mathbb{R}_{\geq 0}}\xspace}%
\newcommand{\Qp}{\ensuremath{\mathbb{Q}_{\geq 0}}\xspace}%
\newcommand{\B}{\ensuremath{\mathbb{B}\xspace}}%

\newcommand{\timoncolor}{brown!20}
\newcommand{\jonnicolor}{red!20}
\newcommand{\ninacolor}{green!20}
\newcommand{\juhacolor}{purple!20}
\newcommand{\miikacolor}{teal!20}

\newcommand{\con}{\ensuremath{C}\xspace}
\newcommand{\cons}{\ensuremath{\mathcal{C}}\xspace}

\tcbset{highlightmath/.style={
  colback=gray!10,
  colframe=black,
  boxrule=0.5pt,
  arc=4pt,
  left=5pt,
  right=5pt,
  top=5pt,
  bottom=5pt,
  sharp corners,
}}

\newtcolorbox{mathbox}{highlightmath}

\usepackage[switch,mathlines]{lineno}
\usepackage{etoolbox}

\newcommand*\linenomathpatch[1]{%
  \cspreto{#1}{\linenomath}%
  \cspreto{#1*}{\linenomath}%
  \csappto{end#1}{\endlinenomath}%
  \csappto{end#1*}{\endlinenomath}%
}
\newcommand*\linenomathpatchAMS[1]{%
  \cspreto{#1}{\linenomathAMS}%
  \cspreto{#1*}{\linenomathAMS}%
  \csappto{end#1}{\endlinenomath}%
  \csappto{end#1*}{\endlinenomath}%
}

\expandafter\ifx\linenomath\linenomathWithnumbers
  \let\linenomathAMS\linenomathWithnumbers
  \patchcmd\linenomathAMS{\advance\postdisplaypenalty\linenopenalty}{}{}{}
\else
  \let\linenomathAMS\linenomathNonumbers
\fi

\linenomathpatch{equation}
\linenomathpatchAMS{gather}
\linenomathpatchAMS{multline}
\linenomathpatchAMS{align}
\linenomathpatchAMS{alignat}
\linenomathpatchAMS{flalign}

\makeatletter
\patchcmd{\mmeasure@}{\measuring@true}{
  \measuring@true
  \ifnum-\linenopenaltypar>\interdisplaylinepenalty
    \advance\interdisplaylinepenalty-\linenopenalty
  \fi
  }{}{}
\makeatother

\title{Locally Consistent K-relations:\\ Entailment and Axioms of Functional Dependence}

\iffalse %
\author{%
    Anonymous author(s)
}
\else
\author{%
Timon Barlag$^1$\and
Miika Hannula$^2$\and
Juha Kontinen$^{3}$\and
Nina Pardal$^4$\and
Jonni Virtema$^5$ \\
\affiliations
$^1$Leibniz Universit\"at Hannover, Germany \,
$^2$University of Tartu, Estonia  \\
$^3$University of Helsinki, Finland \,
$^4$University of Edinburgh, UK \,
$^5$University of Glasgow, UK\\
\emails
barlag@thi.uni-hannover.de,
hannula@ut.ee,
juha.kontinen@helsinki.fi\\
npardal@ed.ac.uk, jonni.virtema@glasgow.ac.uk.
}
\fi

\begin{document}

 \maketitle

\begin{abstract}
Local consistency arises in diverse areas, including Bayesian statistics, relational databases, and quantum foundations, and so does
the notion of functional dependence.
We adopt a general approach to study logical inference in a setting that enables both global inconsistency and local consistency.
Our approach builds upon pairwise consistent families of $K$-relations, i.e, relations with tuples annotated with elements of some positive commutative monoid.
The framework covers, e.g., families of probability distributions arising from quantum experiments and their possibilistic counterparts.
As a first step, we investigate the entailment problem for functional dependencies (FDs) in this setting. Notably, the transitivity rule for FDs is no longer sound, but can be replaced by two novel axiom schemas. We provide a complete axiomatisation for, and establish NL-completeness of, the entailment problem of unary FDs, and demonstrate that even this restricted case exhibits context-dependent subtleties. In addition, we explore when contextual families over the Booleans have realisations as contextual families over various monoids. 
\end{abstract}

\section{Introduction}

Local consistency is an important concept in various areas such as Bayesian statistics, relational databases and quantum foundations. {\color{black}In broad terms,} local consistency refers to a family of partial structures (such as marginal distributions or projections) that agree on their overlapping parts.  
{\color{black}
Local consistency is usually a desirable property to enforce, or at least a minimal requirement for an observed phenomenon. It is often contrasted with global consistency, which demands the existence of a single underlying structure whose projections yield the given family. In practice, local consistency sometimes serves as a proxy for global consistency, especially when the latter is infeasible to check or explicitly maintain.

In relational databases, local consistency becomes an issue when data is spread across multiple tables with overlapping attribute sets. The motivation for using multiple tables  %
is data integrity; decomposing data helps prevent anomalies, e.g., related to updates. %
Notably, consistency in relational databases is typically enforced via integrity constraints that are local in nature. For example, data about a person may be distributed across multiple tables, which we expect to be globally consistent; yet this consistency is usually maintained via primary and foreign keys that act on individual tables or pairs of tables only.

In quantum theory, local consistency is a principle that places an upper bound on the outcomes appearing in physical experiments: any experimental setup must produce a family of probability distributions that satisfies local consistency (a.k.a.\ \emph{no-signalling}).
However, not all locally consistent families are compatible with quantum predictions (see, e.g, ~\cite{pawlowski2009information}) and, vice versa, there are quantum realisable measurement scenarios that induce locally consistent probability distributions that are not globally consistent  \cite{Bell:1964kc,06479255de6a}. This intricate tension between local and global consistency (\emph{contextuality}) has been extensively studied in quantum information theory (see, e.g.,  \cite{DBLP:conf/csl/AbramskyBKLM15,Abramsky2011-ABRTSS,PhysRevA.85.062114}) and has been recognised as a fundamental resource in quantum computation,
providing a key explanation for why quantum systems can outperform classical ones \cite{Howard2014}.

In Bayesian statistics, local consistency is central for the representation of global probability distributions via factorisation. For example, Bayesian networks encode a global distribution as a family of marginal and conditional distributions that must be consistent with each other. Such factorisations are very useful as they facilitate computationally feasible probabilistic and approximate reasoning without the need to directly access or store the  global distribution (see, e.g, \cite{10.5555/1795555}).

The intricate relationship between local and global consistency raises natural questions: under what conditions does local consistency entail global consistency.
The so-called \emph{local-to-global consistency} refers to the situation where the underlying structure (such as a database schema or a collection of variable sets) guarantees that every locally consistent family (of relations or distributions) can be extended to a globally consistent one.
The characterisations of this property are typically formulated in the language of graph theory.
For relational databases it is known that a database schema satisfies the local-to-global consistency property if and only if it forms an acyclic hypergraph \cite{10.1145/2402.322389}. Similar characterisations have been also shown in the context of distributions and multisets \cite{AK21,V62}. More recently, these results have been generalised in the setting of $K$-relations, which are relations whose tuples are annotated with elements from some positive commutative monoid \cite{10.1145/3651608}.
In this case, the assumption of acyclicity of the schema is not always sufficient but a more involved characterisation is required.

A distinct but related area to study is what kind of inferences remain valid when only local consistency is guaranteed. In this setting, classical inference principles that rely on global consistency may fail. A particularly simple but instructive example is the failure of transitivity of functional dependence: even if $A$ functionally determines $B$ and $B$ functionally determines $C$ in two locally consistent views, it does not necessarily follow that $A$ functionally determines $C$ globally.\looseness=-1

\begin{example}\label{ex:new}

\begin{figure}
    \centering
\scalebox{0.9}{
\begin{minipage}[t]{0.25\textwidth}
\begin{tabular}[t]{p{1.4cm} p{1.6cm}}
        \texttt{Patient}&\texttt{Diagnosis}  \\
        \toprule
        P1&Cancer \\
        P2&ALS \\
        \bottomrule
        \end{tabular}   

\vspace{3.35mm}
    
\begin{tabular}[t]{p{1.6cm} p{1.8cm}}
        \texttt{Diagnosis}&\texttt{Treatment}  \\
        \toprule
        Cancer&Chemotherapy \\ 
        ALS&Palliative\\
        \bottomrule
    \end{tabular} 
\end{minipage}

\begin{minipage}[t]{0.25\textwidth}
    
\begin{tabular}[t]{p{1.8cm} p{1.8cm}}
         \texttt{Patient}&\texttt{Treatment}  \\
         \toprule
        P1&Chemotherapy \\
        P1&Palliative \\
        P2&Palliative \\
        \bottomrule
    \end{tabular}
\end{minipage}    
    }
    \caption{Locally consistent family of relations for patients, diagnoses, and treatments.}
    \label{fig:intro-medical}
\end{figure}

Figure~\ref{fig:intro-medical} illustrates three distinct semantic views to hospital data. In the top-left table, each patient is assigned exactly one diagnosis, captured by the functional dependency $\fd{\texttt{Patient}}{\texttt{Diagnosis}}$, and in the bottom-left table, each diagnosis is associated with a preferred treatment according to clinical guidelines, yielding $\fd{\texttt{Diagnosis}}{\texttt{Treatment}}$. The right-hand side table reflects actual clinical practise where doctors have freedom to deviate from standard guidelines, for instance by performing different treatments in sequence when earlier interventions fail. As a consequence, this table does not satisfy $\fd{\texttt{Patient}}{\texttt{Treatment}}$, in violation of transitive closure. The three tables are locally consistent in the sense that any two tables agree on patients, treatments, and diagnoses whenever they record them.
However, the tables are not globally consistent, since there exists no single relation whose projections coincide with all three tables simultaneously. 

More generally, treatments of patients may be irreducible to any finite set of
rules, even when additional variables are introduced, as there are always contextual factors and exceptions that escape formal representation. Nevertheless, we may still wish to retain certain rules as abstract
guidelines in an idealized context. In our example, $\fd{\texttt{Diagnosis}}{\texttt{Treatment}}$ captures a general rule in a context where individual
patients are absent by design.
The
example thus models a situation where real-life complexity forces us to drop
the assumption of global representability (analogously to the lack of global classical probability distribution in quantum physics). \looseness=-1
\end{example}
The tension between local and global consistency plays a central role in Ontology-Based Data Access (OBDA), where reasoning is often performed over locally consistent views induced by ontologies and mappings, even when global consistency of the underlying data is absent or undecidable to check. Prior work~\cite{ecai-ConsoleL14,ISWC-LemboMRST15} has highlighted the practical and theoretical relevance of local reasoning principles in this setting, showing that meaningful query answering and inference can be achieved under local consistency assumptions.
Other prior frameworks that develop the theory of logical inference in connection to local consistency include the logic of local inference of~\cite{kishida:LIPIcs.ICALP.2016.113} %
and quantum team logic by \citeauthor{HYTTINEN} \shortcite{HYTTINEN}, both of which are geared for applications in quantum theory.%

\vspace{2mm}
\noindent\textbf{Our Contributions.} 
In this article, we turn the tables and treat local consistency not as a proxy for global consistency, but as a foundational concept in its own right.
 Analogously to~\cite{10.1145/3651608}, %
 we adopt the useful abstraction of $K$-relations, where $K$ is a monoid, as they uniformly generalise Boolean relations, bags, and probability distributions, while preserving the algebraic structure (addition) needed for marginalising and reasoning about consistency. This algebraic abstraction also highlights prominent characteristics of relations, bags, and distributions, clarifies when and why they diverge, and offers a framework for generalisation to other contexts.
 We also redirect our focus from domains such as information theory and graph theory to 
 logical consequence, and investigate how it is impacted by the shift from global to local consistency. As a test case, we study the logical entailment (i.e., the implication problem) between functional dependencies (FDs) over locally consistent families of $K$-relations. It is worth noting that the assumption of local consistency makes it possible to assign a truth value to an FD over a \emph{family of relations} with shared attributes.
It is important to note that while satisfaction of FDs on any given family of locally consistent $K$-relations depends only on the support of that family, validity of entailment can depend on the monoid, for not all monoids yield the same class of supports of locally consistent $K$-relations.
 
  After defining our theoretical framework, we establish a sound and complete axiomatisation for the entailment of unary FDs with binary CDs in the setting of locally consistent $K$-relations.
  A \emph{context dependency} (CD) is a kind of a dummy FD of the form $\fd{x}{x}$ that only declares that variables in $x$ are locally consistent. 
 Further, we observe that the usual transitivity rule of FDs becomes unsound. Interestingly, the transitivity rule disperses into an infinite collection of weaker rules that we name \emph{chain and cycle rules}. In a globally consistent setting, these rules are all derivable from the transitivity rule, and one instance of a chain rule derives the transitivity rule.
 We then establish that the derivability problem of our axiomatisation for unary FDs is NL-complete in conjunction with binary CDs, and remains in polynomial time with arbitrary CDs. See Table \ref{table:results} for a summary.

\begin{table}
\centering
\scalebox{0.85}{
    \begin{tabular}{@{}l c | c c | l @{}}
\toprule
        {FDs} & {CDs} & {Axioms} & {Complexity} & {Reference}  \\
\midrule
{all} & $\dagger$   & R+A+T  & $\in$DLIN & $\ddagger$ \\
{unary} & $\dagger$   & R+T  & NL & easy observation \\
\midrule
unary & binary & R+CL & NL & Thms. \ref{thm:completeness} \& \ref{thm:CRNL}  \\
unary & binary & single R+CL & NL & proof of Thm. \ref{thm:CRNL}  \\
       \midrule
    {unary} & all & R+CL+CH (s) & $\in$PTIME & Prop. \ref{prop:advanced_chain_rule} \& Thm. \ref{thm:chain_ptime}  \\
       {unary} & ternary & single CH (s) & NL & Prop. \ref{prop:1chain}  \\
\bottomrule
    \end{tabular}
}
	  \caption{Overview of our axiomatisation and complexity results on entailment of inclusion dependencies. 
	R, A, and T refer to the standard Armstrong's axioms of reflexivity, augmentation, and transitivity (see page \pageref{subsec:Contextual_FDs}), and CL and CH refer to the cycle and chain rules (see pages \pageref{lem:sound} and \pageref{prop:advanced_chain_rule}).
	Axiomatisations marked with (s) are sound, others are sound and complete.
	Single means that derivations can use only assumptions and a single application of a proof rule.
	Complexity results refer to completeness, if not otherwise stated.
	$\dagger$: non-contextual setting. $\ddagger$:\protect\cite{Armstrong_Axioms,BeeriB79}.
}
\label{table:results}
\end{table}

Finally, we investigate when a given family of locally consistent relations $\pX$ can be enriched to a family of locally consistent $K$-relations (with the same support).
   \citeauthor{Abramsky13} \shortcite{Abramsky13} showed, via examples and using different terminology, that there are families of locally consistent relations that cannot be expanded to families of locally consistent probability distributions with the same support.
   We extend these examples to a complete characterisation of this extendability property in the case where $K$ is cancellative and $\pX$ is a contextual family over what we refer to as a \emph{hole context sequence}.
  
  Lastly, we show that the set of supports of contextual $\Rp$-families and
contextual $\N$-families coincide. Hence, an axiomatisation for FDs is sound (resp. complete) w.r.t. contextual $\Rp$-families iff it is that for contextual $\N$-families.

\section{Preliminaries}
For a natural number $n$, we write \( [n] := \{0, 1, \dots, n-1\} \).
We fix a countably infinite set $\Var$ of \emph{variables}. 
Given a set $A$, an \emph{assignment} (of $A$) is a function $s$ that maps a finite set $D \subseteq \Var$ of variables to some values (in $A$). We call $D$ the \emph{domain} of $s$ and write $\Dom(s)$. The set of all assignments $s \colon D \to A$ is denoted $\as(D, A)$. %

A \emph{monoid} is an algebraic structure $K=(K,+,0)$, where $+$ is associative and $0$ is the identity element of $+$.  $K$ is \emph{positive} if $a+b=0$ entails $a=0=b$, for all $a,b\in K$, and \emph{(left) cancellative} if $a+b=a+c$ entails $b=c$. We associate each monoid $K$ with its \emph{natural order} $\leq$, defined by $a \leq b \iff \exists c: a+c = b$. The natural order of a monoid is reflexive and transitive, meaning that it is a \emph{preorder}.
In this paper, $K$ will always denote a positive commutative non-trivial %
 monoid (the trivial monoid contains only $0$).
 Commutativity and positivity enforce that the marginalisation of a K-relation is well-defined and that the operations of taking a support and a marginalisation commute (these operations are defined next).

Given a finite set of variables $D \subseteq \Var$ and a finite set $A$, a \emph{$K$-relation} over $D$ is a function $\KX \colon \as(D, A) \to K$.
We write $\Vars(\KX)$ for the set of variables $D$ and $\Dom(R)$ for the set $\as(D, A)$.
The \emph{support} of $\KX$ is the relation $\supp(\KX) \coloneqq \{ s \in \Dom(R) \mid \KX(s) \neq 0\}$. 

For a relation $\KX$ and an element $c \in K$ of a monoid $K$, we denote by $c \KX$ the $K$-relation of support $\KX$ such that $(c \KX)(s) = c$ for all $s \in \KX$.
For two $K$-relations with the same domain $X$, $\KX \colon X \to K$ and $\KY \colon X \to K$, we define the $K$-relation $\KZ = \KX + \KY$ as follows:
$\KZ \colon X \to K$ and $\KZ(s) = \KX(s) + \KY(s)$ for all $s \in X$.
For a set of variables $D'$, the $K$-relation $\KX \restr D'$ is the \emph{marginalisation} of the $K$-relation $\KX$ to the variables in $D'$.
This means that $\KX \restr D'$ is a $K$-relation $\as(D', A) \to K$ s.t.
\[
    (\KX \restr D')(s) := \sum_{\substack{s' \in \dom(\KX) \\ s = s' \restr D'}} \KX(s'),
\]
where $s' \restr D'$ is the restriction of the assignment $s'$ to $D'$ and $\sum$ is the monoid's aggregate sum.
In addition to the general case, we consider the Boolean monoid $\mathbb{B}=(\{0,1\}, \lor,0)$ and the monoids of non-negative reals $\Rp=([0,\infty), +,0)$ and natural numbers $\N=(\N, +,0)$ with their usual addition.

In quantum information theory, a context is a set of attributes that can be measured together. More formally,
a \emph{context} $\con$ is a set of variables, and a \emph{context set} $\cons$ is a downward closed set of contexts (i.e, $\con\in\cons$ and $\con
'\subseteq \con$ implies $\con'\in\cons$).  We often represent a context set \( \cons \) in terms of its maximal elements; for instance, we may write $\cons=\{\{x,y\},\{y,z\}\}$ to denote $\cons=\{\{x,y\},\{y,z\},\{x\},\{y\},\{z\},\emptyset\}$.  
To avoid clutter, we often write a context $\{x_1, \dots ,x_n\}$ as $x_1\dots x_n$. 
We call $\cons$ a \emph{uni-context}, if $\bigcup \cons \in \cons$.

Let $\KX$ and $\KY$ be $K$-relations over sets of variables $D$ and $D'$, respectively. We say that $\KX$ and $\KY$ are \emph{consistent} if $\KX \upharpoonright (D \cap D') = \KY \upharpoonright (D \cap D')  $. A family of $K$-relations $\pX$ is \emph{locally consistent} if all pairs of $K$-relations from $\pX$ are consistent, and \emph{globally consistent} if there exists a $K$-relation $\KX$ such that $\Vars(\KX)=\bigcup_{\KY\in \pX}\Vars(\KY)$ and $\KX\restr \Vars(\KY)=\KY$ for all $\KY\in \pX$.

\begin{definition}[Contextual $K$-families]
A \emph{contextual $K$-family} over a context set $\cons$ is a locally consistent set of $K$-relations $\pX$ containing exactly one $K$-relation $\KX$ with domain $C$, for each $C \in \cons$. The $K$-relation of $\pX$ with domain $C$ is referred to as $\pX_C$. %
The set $\cons$ is the \emph{context set} of $\pX_\cons$. %
\looseness=-1
\end{definition}
The \emph{support} $\supp(\pX)$ of $\pX$ is defined as $\{\supp(\pX_\con)\mid \con \in \cons\}$. 
A contextual $\B$-family is simply called a \emph{contextual family}, and identified with its support. Since $K$ is assumed to be positive, we obtain the following proposition.

\begin{proposition}
   If $\pX$ is a contextual $K$-family, then $\supp(\pX)$ is a contextual family.
\end{proposition}
If $\pX$ is a contextual family and $\alpha \in K$, we set $\alpha \pX$ as  $\{\alpha\pX_\con \mid \con \in \cons\}$. In contrast to the above, if $\pX$ is a contextual family, $\alpha \pX$ is not necessarily a contextual $K$-family (it may violate local consistency; cf. Ex.~\ref{ex:nosup0}).
Let $\pX$ and $\pY$ be contextual $K$-families over a shared context set $\cons$. We set $\pX + \pY$ as $\{\pX_\con + \pY_\con \mid \con \in \cons\}$. Clearly, local consistency is preserved under this operation.
\begin{proposition}\label{prop:closedplus}
    If  $\pX$ and $\pY$ are contextual $K$-families over a shared context set $\cons$, then $\pX + \pY$ is a contextual $K$-family.
\end{proposition}
A contextual family $\pX$ is sometimes identified with the set of its assignments $\set{\pX}=\bigcup_{\con\in \cons}\pX_\con$.
Similarly, a contextual $K$-family $\pY$ can be viewed as a single function $\#:\set{\supp(\pY)}\to K_{\neq 0}$,
where $K_{\neq 0}=\{a \in K\mid a\neq 0\}$. %
A contextual family is called \emph{$K$-realisable} if it is the support of some contextual $K$-family.

\begin{example}\label{ex:nosup}
The contextual family in Fig.~\ref{fig:intro-medical} is locally but not globally consistent (as stated in \Cref{ex:new}). In particular, each pair of tables agree on the projection to their common part. To see why global consistency does not hold, suppose toward a contradiction that there exists a single relation that projects to the three tables in the figure. Let $s$ be an assignment of this relation mapping $(\texttt{Patient}, \texttt{Treatment})$ to (P1, Palliative). Then $s$ must map $(\texttt{Diagnosis}, \texttt{Treatment})$ to (ALS, Palliative), as can be seen from the bottom-left table. Howe, the top-left table entails $s$ maps $(\texttt{Patient}, \texttt{Diagnosis})$ to (P2, ALS), which means that $s$ maps \texttt{Patient} to two different values, a contradiction.
\end{example}

\section{Entailment of Dependence}

Next, we focus on entailment of functional dependencies on contextual $K$-families. We first summarise the well-known Armstrong's axioms for functional dependence in the classical setting \cite{Armstrong_Axioms}. We then develop the theory in the contextual setting. In particular, we give a sound and complete axiomatisation for unary dependencies and establish that the related derivability problem is NL-complete.

\subsection{Classical Setting}
A \emph{functional dependency} (FD)  is a statement of the form $\fd{x}{y}$, where $x,y\subseteq D$ are finite sets of variables. An FD $\fd{x}{y}$ is \emph{satisfied} in a relation $X$ (written $X\models \fd{x}{y}$) if for all $s,s'\in X$: $\bigwedge_{z\in x}s(z)=s'(z)$ entails $\bigwedge_{z\in y}s(z)=s'(z)$. For a set $\Sigma$ of FDs, we write $\Sigma\models \fd{x}{y}$ if for all relations $X$: if $X\models \phi$ for all $\phi \in \Sigma$, then $X\models \fd{x}{y}$. 

In~\cite{Armstrong_Axioms}, it is %
established that the entailment $\Sigma\models \fd{x}{y}$ is completely axiomatised by the following set of three axioms, termed \emph{Armstrong's axioms}:
\begin{mathbox}
\begin{itemize}
\item \textbf{Reflexivity.} %
   \(
   \fd{x}{y}, \quad \text{for } y \subseteq x.
   \)

\item \textbf{Augmentation.} If \( \fd{x}{y} \), then \fd{xz}{yz}, for any \( z \).

\item \textbf{Transitivity.} If \( \fd{x}{y} \) and \( \fd{y}{z} \), then
   \(
   \fd{x}{z}.
   \)
\end{itemize}
\end{mathbox}

We write by $\Sigma \vdash_\mathrm{A}\fd{x}{y}$ to denote that $\fd{x}{y}$ can be derived from $\Sigma$ using these three axioms.

\subsection{FDs in the Contextual Setting}
\label{subsec:Contextual_FDs}

For a set of functional dependencies $\Sigma\cup\{\phi\}$, we write $\Vars(\phi)$ for the set of variables occurring in $\phi$ and $\Vars(\Sigma):= \bigcup\{\Vars(\psi) \mid \psi\in\Sigma\}$. 
We write $\Cont(\Sigma) := \{C \subseteq \Vars(\varphi) \mid \varphi \in \Sigma\}$ for the \emph{set of contexts of $\Sigma$}.

For an FD $\fd{x}{y}$ and a $K$-relation $\KX$ s.t. $x,y\subseteq \Vars(\KX)$, it is natural to set that $\KX\models \fd{x}{y}$, if $\supp(\KX)\models \fd{x}{y}$.
Given a contextual $K$-family $\pX$ over a context set $\cons$ and an FD $\phi$ such that $\Vars(\phi)\in \cons$, we say that $\pX$ \emph{satisfies} $\phi$, written $\pX \models \phi$, if $\pX_{\Vars(\phi)}\models\phi$.
Note that local consistency guarantees that $\pX \models \phi$ iff $\pX_{C}\models\phi$, for any (all, resp.) context $C$ such that $\Vars(\phi)\subseteq C$.
Moreover, $\pX \models \phi$ is undefined, and not a well-formed statement, if $\Vars(\phi)\notin \cons$.
We write $\pX \models \Sigma$, if $\pX \models \phi$, for all $\phi\in\Sigma$.

FDs of the form $\fd{x}{x}$ are called \emph{context dependencies} or CDs. A CD $\fd{x}{x}$ is thus a kind of a dummy FD which only states that $x$ is a part of the context set.
We abbreviate CDs by writing $x$ instead of $\fd{x}{x}$.

A finite set of FDs $\Sigma$ \emph{entails} an individual FD $\phi$ over $K$ (written $\Sigma \models_K \phi$) if for every contextual $K$-family $\pX$ over $\cons$, such that $\Cont(\Sigma \cup\{\phi\}) \subseteq \cons$, $\pX \models \Sigma$ implies $\pX \models \phi$. If $K=\B$, we sometimes write $\Sigma \models \phi$ and say that $\Sigma$ \emph{entails} $\phi$ (without explicitly mentioning $\B$).

It is not difficult to see that reflexivity and augmentation remain sound for entailment in the contextual setting. Transitivity, in turn, is not sound; recall that Fig.~\ref{fig:intro-medical} (and its support) showcases a counterexample.
 It is not difficult to see that the following weaker form of the transitivity rule, termed \emph{contextual transitivity}, is sound in the contextual setting.

\begin{mathbox}
 \textbf{Contextual transitivity.} If \( \fd{x}{y} \), \( \fd{y}{z} \) and $\fd{xyz}{xyz}$, then
   \(
   \fd{x}{z}.
   \)
\end{mathbox}

\begin{remark}\label{remark:entailment}
The statement $\Sigma \models \phi$ intuitively reads:  \emph{Every contextual family, where $\Sigma\cup\{\phi\}$ are well-formed formulae and that satisfies every FD in $\Sigma$ also satisfies $\phi$.}
Hence, compositionality of entailment: $\Sigma \models \phi$ and $\Sigma \cup \{\phi\} \models \psi$ implies $\Sigma \models \psi$, is guaranteed only when $\Vars(\phi)\in \Cont(\Sigma \cup \{\psi\})$.
A counterexample is $\Sigma=\{\fd{x}{y},\fd{y}{z}\}, \phi = \fd{xyz}{xyz},\psi=\fd{x}{z}$. 
One way to preserve compositionality is to define a variant of entailment that is parametric to the context set. That is,
for a set of FDs $\Sigma\cup\{ \phi \}$ and a context set \cons such that $\Cont(\Sigma\cup\{ \phi \})\subseteq\cons$, we could define  $\Sigma \models^\cons_K \phi$ if $\pX \models \Sigma$ implies $\pX \models \phi$ for every contextual $K$-family $\pX$ over $\cons$. As our goal is to obtain axiomatisations of the entailment relation that do not utilise context sets in the axioms, we do not follow this approach in this paper.
\end{remark}

 We note the following simple property, which follows from the facts that (i) the support of a contextual $K$-family is a contextual family, and (ii) $\pX\models \phi \Leftrightarrow \supp(\pX)\models \phi$ holds for all contextual $K$-families $\pX$ and FDs $\phi$.
\begin{proposition}\label{prop:simple}
 $\Sigma \models \phi$ implies $\Sigma \models_{K} \phi$, for any set of FDs $\Sigma\cup\{\phi\}$ and positive commutative monoid $K$. 
\end{proposition}
It is not self-evident when the converse holds. Example \ref{ex:nosup0} in Section \ref{sec:realisability} showcases that there are contextual families that are not realised as supports of contextual $\Rp$-families.  

The notion of a {formal proof} in the contextual setting has a subtle feature; the intermediate derivations should not introduce new contexts.
Otherwise, for instance, it would be possible to derive $\fd{x}{z}$ from $\fd{x}{y}$ and $\fd{y}{z}$ using reflexivity (to obtain $\fd{xyz}{xyz}$) and contextual transitivity (to obtain $\fd{x}{z}$). In contrast, the question of logical entailment of $\fd{x}{z}$ by $\fd{x}{y}$ and $\fd{y}{z}$ has been formulated over contextual families having contexts $xy,yz,xz$ but not necessarily the context $xyz$.

\begin{definition}[Derivation]\label{def:derivation}
Given a set of inference rules $\mathcal{S}$ and a set of FDs $\Sigma \cup\{\phi\}$, we say that $\phi$ is \emph{derivable} from $\Sigma$ by $\mathcal{S}$, written $\Sigma \vdash_\mathcal{S} \phi$, if there is a finite sequence of FDs $\psi_0, \dots ,\psi_{n-1}$ such that 
\begin{enumerate}
    \item $\psi_{n-1}=\phi$
    \item for $i\in [n]$, $\psi_i$ is from $\Sigma$ or obtained from $\{\psi_0, \dots ,\psi_{i-1}\}$ by using one of the rules of $\mathcal{S}$.
    \item\label{it:new} for $i\in [n]$, $\Vars(\psi_i)\in \Cont(\Sigma \cup\{\phi\})$.
\end{enumerate}
Specifically, \Cref{it:new} is the new criterion introduced for the contextual setting.
As usual, we call $\mathcal{S}$ \emph{complete} if $\Sigma\models \phi$ entails $\Sigma \vdash_\mathcal{S} \phi$, and \emph{sound} if  $\Sigma\vdash_\mathcal{S}  \phi$ entails $\Sigma \models \phi$.
\end{definition}
\begin{remark}\label{rem:composition}
    Analogous to Remark~\ref{remark:entailment}, derivations for $\Sigma \vdash \phi$ and $\Sigma \cup \{\phi\} \vdash \psi$ can be surely composed to a derivation of $\Sigma \vdash \psi$ only when $\Vars(\phi)\in \Cont(\Sigma \cup\{\psi\})$.
  Again, one could regain compositionality by making the proof system parametric to the context set (as described in Remark~\ref{remark:entailment}) or by requiring that $\Vars(\phi)\in\Cont(\Sigma)$ whenever $\Sigma \vdash \phi$.
\end{remark}

We write $\textrm{CA}$ for the modified set of Armstrong's axioms, where the transitivity rule is replaced with the contextual transitivity rule, and $\vdash_{\textrm{CA}}$ for the related entailment relation.

We say that an axiom $\Sigma\vdash \varphi$ is uni-context, if $\bigcup \Cont(\Sigma\cup\{ \varphi\})\in \Cont(\Sigma\cup\{\varphi\})$. A set of axioms is uni-context, if each of its axioms is. E.g., $\textrm{CA}$ is uni-context, whereas the transitivity rule is not.
We conclude this section by establishing that $\vdash_{\textrm{CA}}$ is the strongest uni-context axiomatisation.

\begin{proposition}\label{prop:CAtoA}
Let $\Sigma\cup\{\fd{x}{y}\}$ be a set of FDs and $z$ a set variables such that $\Vars(\Sigma\cup\{\fd{x}{y}\})\subseteq z$. Then
\[
\Sigma\cup\{\fd{z}{z}\}\vdash_{\mathrm{CA}} \fd{x}{y} \text{ iff } \Sigma\vdash_{\mathrm{A}} \fd{x}{y}.
\]
\end{proposition}
\begin{proof}
The "$\Rightarrow$" direction is immediate, for every instance of the contextual transitivity rule can be replaced by an instance of the transitivity rule and \fd{z}{z} is derivable (in the non-contextual setting of $\vdash_{\mathrm{A}}$) by the reflexivity rule.

For "$\Leftarrow$", assume $\Sigma\vdash_{\mathrm{A}} \fd{x}{y}$. It follows by inspecting  \cite[Corollary 2]{BeeriB79} that there exists a proof of $\Sigma\vdash_{\mathrm{A}} \fd{x}{y}$ that only uses variables in $z$. This proof can be modified into a proof of $\Sigma\cup\{\fd{z}{z}\}\vdash_{\mathrm{CA}} \fd{x}{y}$ by replacing each occurrence of the transitivity rule by the corresponding contextual transitivity rule together with an application of the reflexivity rule to obtain \fd{z'}{z'} for the required $z'\subseteq z$ to apply the contextual transitivity rule.
\end{proof}

\begin{restatable}{theorem}{isolated}\label{thm:isolated}   
Let $\textrm{AX}$ be a set of sound uni-context axioms such that $\mathrm{CA} \subseteq \textrm{AX}$, and $\Sigma \cup \{\fd{x}{y}\}$ a set of FDs. Then $\Sigma \vdash_\textrm{AX} \fd{x}{y}$ iff $\Sigma \vdash_\textrm{CA} \fd{x}{y}$.
\end{restatable}

Next, we establish a sound and complete axiomatisation for unary FDs, and as a side product obtain that in this case the converse of \Cref{prop:simple} holds (cf. Corollary \ref{cor:simple}).

\subsection{Unary FDs with Binary Contexts}
An FD $\fd{x}{y}$ is \emph{unary} if $x$ and $y$ are individual variables. 
A CD $x$, where $x$ consists of two variables, is a \emph{binary} CD.
In the rest of this section, we restrict to the case, where $x$, $y$, etc. refer to individual variables, instead of sets of variables. 
Surprisingly, to obtain a complete axiomatisation in the setting of unary FDs and binary CDs, we only need to append reflexivity with one extra axiom schema, stating that every cycle of unary functional dependencies can be inverted.
This property, termed the \emph{cycle rule}, is presented below.
We write $x_1 \to x_2 \to \dots \to x_{n-1} \to x_n$ as a shorthand for $\fd{x_1}{x_2}, \fd{x_2}{x_3}, \dots ,\fd{x_{n-1}}{x_n} $.
\begin{mathbox}
    \label{box:cycle_rule}
    \textbf{Cycle rule.} For $k\geq 1$, if 
    \(x_1 \to x_2 \to \dots \to x_k \to x_1\),
    then $\fd{x_1}{x_k}$.
\end{mathbox}

An alternative interpretation of the cycle rule states that within every strongly connected component of the dependency graph (consider unary FDs as directed edges between variables, see Def.~\ref{def:dep_graph} for a definition) the edge relation is symmetric.   
The soundness proof of the axiom relies on the assumption of local consistency, as we shall shortly see.

\begin{lemma}\label{lem:sound}
    Let $K$ be a positive commutative monoid.
    The cycle rule is sound for contextual $K$-families, i.e., $x_1 \to x_2 \to \dots \to x_k \to x_1 \models_K \fd{x_1}{x_k}$, for all variables $x_1, \dots, x_k$ where $k\geq 1$.
\end{lemma}
    \begin{proof}
    We prove the case $K=\B$; the general case then follows from Prop.~\ref{prop:simple}.
    Assuming the rule is not sound, 
    we construct a contradiction with the assumption that the contextual family is finite.
    The contradiction is depicted in Fig.~\ref{fig:cycle}.
    Let $\pX$ be a contextual family that satisfies $x_1 \to x_2 \to \dots \to x_k \to x_1$ but does not satisfy $\fd{x_1}{x_k}$. For a relation over two variables $x_i,x_j$, write $(a_{i,\ell},a_{j,\ell'})$ to denote the assignment $x_i\mapsto a_{i,\ell}$ and $x_j\mapsto a_{j,\ell'}$.
    By hypothesis, $\pX_{x_k,x_1}$ contains two assignments $(a_{k,1}, a_{1,1})$ and $(a_{k,2}, a_{1,1})$, where $a_{k,1}\neq a_{k,2}$.
    By local consistency and the FD $x_{k-1} \to x_k$, we observe $\pX_{x_{k-1},x_k}$ then contains two assignments $ (a_{k-1,1}, a_{k,1})$ and $ (a_{k-1,2}, a_{k,2})$, where $a_{k-1,1}\neq a_{k-1,2}$.
    By induction we obtain that $\pX_{x_1,x_2}$ contains also two assignments of the form $ (b_{1,1}, a_{2,1})$ and $ (b_{1,2}, a_{2,2})$, where $b_{1,1}\neq b_{1,2}$ and $a_{2,1}\neq a_{2,2}$.

    Then, choosing $a_{1,2}\in \{b_{1,1},b_{1,2}\}\setminus \{a_{1,1}\}$, local consistency and the FD $\fd{x_k}{x_1}$ entail that $\pX_{x_k,x_1}$ contains an assignment of the form $(a_{k,3},a_{1,2})$, where $a_{k,3}\notin \{a_{k,1},a_{k,2}\}$. Now similarly to the above, by induction $\pX_{x_1,x_2}$ contains an assignment of the form $ (b_{1,3}, a_{2,3})$, where $b_{1,3}\notin \{b_{1,1},b_{1,2}\}$. In particular, the elements $b_{1,1},b_{1,2},b_{1,3}$ are distinct. Then, we choose $a_{1,3}\in \{b_{1,1},b_{1,2},b_{1,3}\}\setminus \{a_{1,1},a_{1,2}\}$, and continue as before. Thus, we see that $\pX$ must contain relations whose supports are infinite, which contradicts the definition. Therefore we conclude that the rule is sound. %
\begin{figure}
    \centering
    \begin{tabular}[t]{ccccc}
\begin{tabular}[t]{cc}
        $x_1$ & $x_2$ \\
        \toprule
        $b_{1,1}$ & $a_{2,1}$ \\
        $b_{1,2}$ & $a_{2,2}$ \\\hdashline
        $b_{1,3}$ & $a_{2,3}$ \\\hdashline
        \vdots & \vdots \\
        \bottomrule
    \end{tabular}
    &\hspace{-3.5mm}
    \begin{tabular}[t]{c}
        \\[3ex] %
        $\dots$
    \end{tabular}
&\hspace{-3.5mm}
\begin{tabular}[t]{cc}
        $x_{k-1}$ & $x_k$ \\
        \toprule
        $a_{k-1,1}$ & $a_{k,1}$ \\
        $a_{k-1,2}$ & $a_{k,2}$ \\\hdashline
        $a_{k-1,3}$ & $a_{k,3}$\\\hdashline
        \vdots & \vdots \\
        \bottomrule
    \end{tabular}
    &
    \,
        \begin{tabular}[t]{cc}
        $x_k$ & $x_1$ \\
        \toprule
        $a_{k,1}$ & $a_{1,1}$ \\
        $a_{k,2}$ & $a_{1,1}$ \\\hdashline
        $a_{k,3}$ & $a_{1,2}$\\\hdashline
        \vdots & \vdots \\
        \bottomrule
    \end{tabular}
    \end{tabular}
    \caption{Counterexample for the falsification of the cycle rule.}
    \label{fig:cycle}
\end{figure}
\end{proof}

Notice that the previous proof constructs a counterexample showing that the cycle rule is not sound if  relations with infinite support are allowed. 
Restricting to unary FDs and binary CDs, we next show that the cycle rule together with reflexivity forms an infinite sound and complete proof system. 
We write $\mathcal{CR}$ for this set of axioms.

\begin{restatable}{theorem}{completeness}\label{thm:completeness}   
For every positive commutative monoid $K$, $\Sigma\models_K \phi$ iff $\Sigma \vdash_{\mathcal{CR}} \phi$,
    assuming $\Sigma \cup \{\phi\} $ is a set of unary FDs and binary CDs.
    Specifically, the cycle rule and reflexivity are sound and complete for unary FDs and binary CDs. 
\end{restatable}

    As the same axioms characterise entailment of unary FDs with respect to both $\B$ and $K$, we obtain the converse of \Cref{prop:simple} as a direct corollary.
\begin{corollary}\label{cor:simple}
     Let $K$ be any positive commutative monoid.
 Now, $\Sigma \models \phi$ if and only if $\Sigma \models_{K} \phi$, for any set of unary FDs and binary CDs $\Sigma\cup\{\phi\}$. 
\end{corollary}

In the classical setting, there is a deterministic linear time algorithm to check whether a given set of FDs entails a given FD \cite{BeeriB79}. 
In the setting of unary FDs with binary contexts, it turns out that the corresponding problem is $\mathrm{NL}$-complete.

Before proving this connection, we introduce the \emph{dependency graph} of a set of FDs.
This dependency graph is essentially the resulting graph when interpreting FDs as directed edges between variables.

\begin{definition}\label{def:dep_graph}
    Let $\Sigma$ be a set of unary FDs.
    The \emph{dependency graph} of $\Sigma$ is the graph $(\Vars(\Sigma), \{(x, y) \mid x \to y \in \Sigma\})$.
\end{definition}

\begin{theorem}\label{thm:CRNL}
    The problem of deciding whether $\Sigma \vdash_{\mathcal{CR}} \phi$ when given a set $\Sigma \cup \phi$ of unary FDs and binary CDs is $\mathrm{NL}$-complete.
    Moreover, the corresponding entailment problem $\Sigma \models_{K} \phi$ is $\NL$-complete, for any commutative monoid $K$.
\end{theorem}
\begin{proof}
By the previously established  completeness theorem (Theorem \ref{thm:completeness}), it suffices to prove the result for $\vdash_{\mathcal{CR}}$. 

    Note first that any FD that can be derived from $\Sigma$ by repeated applications of the cycle rule can also be derived from $\Sigma$ by a single use of the cycle rule.
    Observe that $\fd{x_1}{x_n}$ follows from $\Sigma$ via the cycle rule if and only if there is a cycle containing the edge $(x_n, x_1)$ in the dependency graph of $\Sigma$.
    For a contradiction, assume that there is an FD $\fd{y_1}{y_2}$, such that $\fd{y_1}{y_2}$ derives from $\Sigma$ via the cycle rule, $\fd{x_1}{x_n}$ derives from $\Sigma \cup \{\fd{y_1}{y_2}\}$ via the cycle rule, but $\fd{x_1}{x_n}$ does not derive from $\Sigma$ via the cycle rule.
    Since $\fd{x_1}{y_n}$ derives from $\Sigma \cup \{\fd{y_1}{y_2}\}$ but not from $\Sigma$, there is a cycle containing the edges $(x_n, x_1)$ and $(y_1, y_2)$ in the dependency graph of $\Sigma \cup \{\fd{y_1}{y_2}\}$ and since $\fd{y_1}{y_2}$ derives from $\Sigma$, there is a cycle containing the edge $(y_2, y_1)$ in the dependency graph of $\Sigma$.
    This means, however, that there is a cycle containing the edge $(x_n, x_1)$ path in the dependency graph of $\Sigma$, since the edge $(y_1, y_2)$ can be replaced by the path from $y_1$ to $y_2$ in the respective cycle containing the edge $(y_2, y_1)$.
    Therefore, $\fd{x_1}{x_n}$ can be derived from $\Sigma$ via a single use of the cycle rule.    
    
    Now we can reduce the problem problem of deriving $\fd{x_1}{x_n}$ via the cycle and reflexivity rule to the well-known $\mathrm{NL}$-complete graph reachability problem~\cite{DBLP:books/daglib/0072413}: for a graph $G$ and two nodes $s, t$, $(G, s, t) \in \mathrm{REACH}$ if and only if there is a path from $s$ to $t$ in $G$.
    We do this via the reduction $(\Sigma, \fd{x_1}{x_n}) \mapsto (G_{\Sigma'}, x_n, x_1)$, where $G_{\Sigma'}$ is the dependency graph of $\Sigma$ extended by all unary FDs and binary CDs derivable from the reflexivity rule.

    We show hardness similarly: 
    When asked if there is an $s-t$ path in a graph $G$, we interpret $G$ as the dependency graph of a set of FDs $\Sigma$ and ask whether $\fd{t}{s}$ follows from $\Sigma$ by an application of the cycle rule.
    Formally, this means that we map an instance $((V, E), s, t)$ of $\mathrm{REACH}$ to $(\{(\fd{u}{v} \mid (u, v) \in E\} \cup \{uv \mid (u, v) \in E\}, \fd{t}{s})$.
\end{proof}

\section{$K$-realisability of Contextual Families}\label{sec:realisability}
In this section, we study the $K$-realisability of contextual families.
In particular, we focus on the special case of what we call hole context sequences, and additionally we consider the relationship between realisability in $\Rp$ and $\N$.

\subsection{Hole Context Sets}
It is natural to ask whether or not 
each contextual family is the support of some contextual $K$-family.
This turns out not to be the case for $K=\Rp$. A counterexample (in our terminology, a contextual family over the context set $\{ab,ab',a'b,a'b'\}$) is given in \cite[Proposition 9.1]{Abramsky13}. Our running example %
gives rise to another counterexample, as we argue next.

\begin{example}\label{ex:nosup0}
Consider the contextual family in Fig.~\ref{fig:intro-medical}. For a contradiction, suppose this contextual family is the support of some contextual $\Rp$-family, represented as a function $\# \colon S \to \mathbb{R}_{>0}$, where $S$ is the set of assignments in Fig.~\ref{fig:intro-medical}.
By local consistency, $\#$ uniquely extends to assignments on individual variables. 
Now, abbreviating attribute names and values, we obtain 
\begin{align*}
&\#(\texttt{P}\mapsto \text{P1})=\#(\texttt{D}\mapsto \text{Cancer})
= \#(\texttt{T}\mapsto \text{Chemo})\\
< &\#(\texttt{P,T}\mapsto \text{P1, Chemo})
+\#(\texttt{P,T}\mapsto \text{P1, Palli})\\
=& \#(\texttt{P}\mapsto \text{P1}),
\end{align*}
where the inequality follows since each value of $\#$ is strictly positive.
 Thus the contextual family of Fig.~\ref{fig:intro-medical} cannot be obtained from the support of any contextual $\Rp$-family.
\end{example}

Next, we completely characterise which contextual families are supports of some contextual $K$-family for cancellative monoids $K$ and context sets that form a {hole} (defined below). We show that a contextual family $\pX$ is the support of some contextual $K$-family if and only if its associated \emph{overlap projection graph} is \emph{edge-cyclic}; both terms are defined below.

In what follows, we treat addition and subtraction of natural numbers $i\in [n]$ modulo \( n \). 
Furthermore, collections of contexts will be formally represented as sequences. Since contextual families are defined over context sets, we identify the two whenever convenient.
We say that a sequence of contexts $\cons=(C_0, \dots ,C_{n-1})$, $n\geq 3$, is a \emph{hole} if and only if for all $i,j\in \{0, \dots, n-1\}$, the contexts $C_i$ and $C_{j}$ intersect if and only if $i\in \{j+1,j,j-1\}$. 
A context set is a \emph{hole} if it admits an ordering that forms a hole.
For the remainder of this section, if the context sequence ${\cons}=(C_0, \dots ,C_{n-1})$ is implicitly understood, we write $\pX_i$ as a shorthand for $\pX_{C_i}$, given a contextual $K$-family $\pY$ over $\cons$.

Let $\pX$  be a contextual family over a hole  $\cons=(C_0, \dots ,C_{n-1})$. The \emph{overlap projection graph} of $\pX$, denoted $\opgraph{\pX,\cons}$, is the digraph s.t.:
\begin{itemize}
    \item The vertices consist of the restrictions $s\restr (C_i\cap C_{i+1})$, for each $s\in \pX_i$, $i\in [n]$, and
    \item There is a directed edge from $s\restr (C_i\cap C_{i-1})$ to $s\restr (C_i\cap C_{i+1})$, for each $s\in \pX_i$, $i\in [n]$. 
\end{itemize}
An edge from a vertex $u$ to a vertex $v$ generated by an assignment $s$ is denoted by $u\ledge{s} v$.
 A digraph $G$ is called \emph{edge-cyclic} if each edge $(u,v)\in E$ belongs to some directed cycle (equivalently, if its every connected component is strongly connected).%

 We use our running example to demonstrate this concept.
\begin{example}\label{ex:nosup2}
    Consider again the contextual family in \Cref{fig:intro-medical}. The sequence $\mathcal{C}=(C_1, C_2, C_3)$,  where $C_1=\{\texttt{Patient}, \texttt{Diagnosis}\}$, $C_2=\{\texttt{Diagnosis}, \texttt{Treatment}\}$, and $C_3=\{\texttt{Treatment}, \texttt{Patient}\}$,
    forms a hole.
    The corresponding overlap projection graph, depicted in \Cref{fig:cycles}, has $6$ vertices (one for each value) and $7$ edges (one for each assignment). 
    In the figure, we represent a vertex, that is, a restriction $s\upharpoonright \{x\}$, simply as the unique value $s(x)$ it corresponds to. 
    The vertices $\text{P1}$, $\text{Cancer}$, and $\text{Chemotherapy}$
    form a cycle, and so do $\text{P2}$, $\text{ALS}$, and $\text{Palliative}$. Each edge belongs to some cycle, except for the edge from $\text{Palliative}$ to $\text{P1}$.  The graph is thus not edge-cyclic.
\end{example}

\begin{figure}
    \centering

\begin{tikzpicture}[
  every node/.style={rectangle, draw, minimum size=.5cm},
  >=stealth,
  every path/.style={<-, thick, >=stealth, line width=1pt},
  arrows={-stealth[scale=1.5]}
]

  \node (A) at (0, 2) {P2};
  \node (B) at (2.5, 2.5) {Palliative};
  \node (C) at (5, 2) {ALS};

  \node (D) at (0, 3) {P1};
  \node (E) at (2.5, 3.5) {Chemotherapy};
  \node (F) at (5, 3) {Cancer};

  \draw (A) -- (B);
  \draw (B) -- (C);
  \draw (C) -- (A);

\draw (D) -- (B);

  \draw (D) -- (E);
  \draw (E) -- (F);
  \draw (F) -- (D);
\end{tikzpicture}
    \caption{Example of an Overlap Projection Graph} %
    \label{fig:cycles}
\end{figure}

Then, $K$-realisable contextual families over holes can be characterised via the concept of edge-cyclicity. The proof, which uses auxiliary concepts and results, is deferred to the supplementary material.
\iffull

\begin{proof}
    We need to prove that $w\pX$ is locally consistent. 
    Suppose the cycle has length $\ell$, as in \eqref{eq:simple}. Since $\opgraph{\pX,\cons}$ is cyclically $n$-partite (\Cref{lem:easy}), $\ell$ must be a multiple of $n$; that is, $n\cdot m=\ell$ for some positive integer $m$.
    Fix $i \in [n]$.
    One observes that 
    \[
   |\pX_i\restr (C_{i-1} \cap C_{i})|=|\pX_i| =|\pX_i\restr (C_i \cap C_{i+1})|=m.
    \]
    Moreover, each element from $\pX_i\restr (C_i \cap C_{i+1})$ appears in $\pX_{i+1}\restr (C_i \cap C_{i+1})$, and conversely each element of $\pX_{i+1}\restr (C_i \cap C_{i+1})$ appears in $\pX_i\restr (C_i \cap C_{i+1})$. We conclude that $(w\pX)_i
 \restr (C_i \cap C_{i+1}) = (w\pX)_{i+1} \restr (C_i \cap C_{i+1})$. If $i$ and $j$ are not adjacent, then $C_i \cap C_j= \emptyset$. For this case it suffices to note that $(w\pX)_i\restr \emptyset$ consists of the empty assignment associated with the sum ${w+\dots +w}$ of length $m$. Thus $w\pX$ is a contextual $K$-family. It follows by positivity of $K$ that $\supp(w\pX)=\pX$.
\end{proof}
\else\fi

\begin{restatable}{theorem}{nosup}\label{thm:nosup}   
    Let $K$ be a
    positive commutative cancellative monoid and
      $\pX$ a contextual family over a hole  $\cons=(C_0, \dots ,C_{n-1})$. Then, $\pX$ is the support of some contextual $K$-family if and only if $\opgraph{\pX,\cons}$ is edge-cyclic.
\end{restatable}

\noindent Consider our running example in light of this theorem:
\begin{example}\label{ex:nosup3}
Example \ref{ex:nosup2} and Thm.~\ref{thm:nosup} imply that the contextual family in Fig.~\ref{fig:intro-medical} is not $K$-realisable, for any cancellative monoid $K$. Since $\Rp$ and $\mathbb{N}$ are cancellative, the family in question does not arise as the support of any locally consistent family of distributions or multisets.
However, after adding $(\texttt{Diagnosis, Treatment})\mapsto (\text{Cancer, Palliative})$, the extended overlap projection graph is edge-cyclic. Then, Thm.~\ref{thm:nosup} implies that the obtained contextual family is $K$-realisable. In the case of distributions and multisets, we can introduce uniform weights to each table (in such a way that the total weights of the tables match one another).
\end{example}

\subsection{Contextual $\Rp$-realisability}
Let us next briefly consider the special case of $\Rp$-realisability.
As a consequence of the previous section, we observe that a contextual $\Rp$-family decomposes to cycles over 
$\opgraph{\pX,\cons}$.
Moreover, we establish that a contextual family is the support of some contextual $\Rp$-family if and only if it is the support of some contextual $\N$-family.
\begin{proposition}\label{prop:combination}
    Each  contextual $\Rp$-family over a hole $\cons$ is a non-negative combination of cycles in $\opgraph{\pX,\cons}$.
\end{proposition}
\begin{proof}
    Let $\pX$ be a contextual $\Rp$-family, with a multiplicity mapping $\#\colon \set{\supp(\pX)}\to \mathbb{R}_{>0}$. Assume $\pX$ is non-empty; otherwise the claim holds trivially. \Cref{thm:nosup} entails that there is a non-empty set of assignments $S\subseteq \set{\supp(\pX)}$ forming a cycle in $\opgraph{\pX,\cons}$. 
 Let $w=\min\{\#(s)\mid s \in S\}$. Define the updated family $\pX'\coloneqq \pX-wS$, where subtraction is defined in the obvious way. Then, $\pX'$ is a contextual $\Rp$-family.
 We repeat this process recursively, setting $\pX \coloneqq \pX- wS$ at each step, until the family becomes empty. Since at each step $S$ is removed from the support, no cycle $S$ is used more than once, meaning that the process will terminate.
 This leads to a decomposition of the form $\pX=w_1S_1 + \dots +w_\ell S_\ell$.
\end{proof}

This concludes our discussion of context sets that form a hole. A natural follow-up question is whether a similar graph-theoretic characterisation exists for $\Rp$-realisability over arbitrary context sets. The next example suggests that the answer is negative.
We now write $x_1 \dots x_n\mapsto a_1\dots a_n$ for the assignment that maps $x_1, \dots ,x_n$ respectively to $ a_1,\dots ,a_n$.
\begin{example}
Consider the contextual family represented in \Cref{fig:two}. We note that restricted to $\{ab,bc,ca\}$ the corresponding overlap projection graph has an edge cover, and thus the restriction is
$\Rp$-realisable by \Cref{thm:nosup}. Similar reasoning applies to the restriction on $\{ab',b'c,ca\}$. However, over the full context set, $\Rp$-realisability does not hold anymore.
Assume toward a contradiction that this is not the case.  Let $S$ be the set of  assignments represented in \Cref{fig:two}.
Then, %
there is a mapping $\#\colon S\to \mathbb{R}_{>0}$ which forms a contextual $\Rp$-relation and is such that $\#(bc\mapsto 01)=m> 0$. By local consistency, $\#(ca\mapsto 10)=n\geq m$, and furthermore $\#(ca\mapsto 10)=\#(ab'\mapsto 00)=\#(b'c\mapsto 00)=\#(ca\mapsto 01)=\#(ab\mapsto 11)=\#(bc\mapsto 11)$. Thus 
\[\#(bc\mapsto 01)+\#(bc \mapsto 11) =n+m \neq n = \#(ca\mapsto 10),
\]
violating local consistency.
    
\begin{figure}
    \centering
\[
    \begin{array}[t]{cccccc}
    \begin{array}[t]{cc}
        a&b  \\
        \toprule
        0&0 \\ %
        1&1\\
        \bottomrule
        \end{array}
            \hspace{1mm}
    &
        \begin{array}[t]{cc}
        b&c  \\
        \toprule
        0&0\\
        0&1\\
        1&1\\
        \bottomrule
    \end{array}
            \hspace{1mm}
            \hspace{1mm}
    &
    
    \begin{array}[t]{cc}
         c&a  \\
         \toprule
        0&1\\
        1&0\\
        \bottomrule
    \end{array}
            \hspace{1mm}
    &
        \begin{array}[t]{cc}
        a&b'  \\
        \toprule
        0&0\\
        1&1\\
        \bottomrule
        \end{array}
            \hspace{1mm}
    &
        \begin{array}[t]{cc}
        b'&c  \\
        \toprule
        0&0\\ 
        1&1\\
        \bottomrule
    \end{array}
    \end{array}
    \]
            \caption{Non-realisable contextual family. \label{fig:two}}
\end{figure}
\end{example}

Nevertheless, we can show that $\Rp$-realisability coincides with $\N$-realisability.
\begin{proposition}\label{prop:char}
    A contextual family $\pX$ is the support of some contextual $\Rp$-family if and only if it is the support of some contextual $\N$-family.
\end{proposition}
\begin{proof}
    It suffices to prove the ``only-if'' direction. For each contextual family $\pX$, there exists a system of linear inequalities $\mathbf{A}\mathbf{x}\geq \mathbf{b}$ (with rational coefficients) that has a (real) solution $\mathbf{x}$ if and only if $\pX$ is the support of some contextual $\Rp$-family. Namely, the variables in $\mathbf{x}$ represent the ``weights'' of the assignments in $\pX$, and the inequalities capture the local consistency criterion and the non-negativity of these ``weights''.  
    Furthermore, it is well-known that if a system of the described form has a real solution, then it has a rational solution \cite{0090562}. Hence $S$ must be the support of some contextual $\Qp$-relation $\pX$, which can be transformed into a contextual $\N$-relation by multiplying the weight of each assignment by a sufficiently large constant. 
\end{proof}

\begin{corollary}\label{cor:RtoN}
An axiom system for FDs is sound (complete, resp.) with respect to $\models_{\Rp}$ if and only it is that for $\models_\N$.
\end{corollary}

\section{A Path to More General Axiomatisations}\label{sec:path}
We have shown a sound and complete axiomatisation for FDs restricted to unary FDs and binary CDs. We also showed that the system consisting of reflexivity, augmentation, and contextual transitivity is the strongest uni-context axiomatisation. Next we scout further avenues for more general axiomatisations, and present some preliminary results.

\subsection{Unary FDs with Ternary Contexts}
For unary FDs with ternary CDs, the interaction is significantly more intricate. We have already observed that transitivity can fail: from $\fd{x_1}{x_2}$ and $\fd{x_2}{x_3}$, it does not necessarily follow that $\fd{x_1}{x_3}$. However, if we additionally assume $\fd{c}{x_3}$, $cx_1x_2$, $cx_2x_3$, and $cx_1x_3$, then $\fd{x_1}{x_3}$ does follow.
The next rule  captures this non-trivial interaction between unary FDs and ternary CDs (the above case is its instance with $n = 3$ and $c_i = c$ for all $i \in [n]$).

\begin{mathbox}
    \label{box:contextual_chain_rule}
    \textbf{Contextual chain rule.} If
    \begin{enumerate}
\item\label{item:fd1}
    $x_1 \to x_2 \to \dots \to x_{n-1} \to x_n,$
\item\label{item:fd2}
$\fd{c_1}{ x_n},\fd{c_2}{ x_n}, \dots ,\fd{c_{n-1}}{x_n}$,
\item\label{item:c1} $x_1c_1x_n$,
    \item\label{item:c2} $x_1c_1x_2,c_1x_2c_2,x_2c_2x_3,\dots ,x_{n-2}c_{n-2}x_{n-1},$\\
    $c_{n-2}x_{n-1}c_{n-1},x_{n-1}c_{n-1}x_n$, and
    \item\label{item:c3} $c_1c_2x_n, c_2c_3x_n, \dots ,c_{n-2}c_{n-1}x_n$,
\end{enumerate}
then $\fd{x_1}{ x_n}$.
\end{mathbox}
It is easy to check that the contextual transitivity rule, applied to single variables, is obtained from the above rule by setting $n=3$, $x_1=x$, $x_2=y$, $x_3=z$, $c_1=y=c_2$.

\begin{proposition}\label{prop:advanced_chain_rule}
Let $K$ be a positive commutative monoid.
The contextual chain rule is sound for $\models_K$.
\end{proposition}
\begin{proof}
We prove the case $K=\B$; the general case then follows from Prop.~\ref{prop:simple}.
    Let $\pX$ be a contextual family such that the assumptions of the contextual chain rule hold.
    Suppose toward a contradiction that $\pX$ does not satisfy $\fd{x_1}{x_n}$. Then $\pX$ contains two assignments $s$ and $s'$ in context $x_1c_1x_n$ such that $s(x_1)=s'(x_1)$ and $s(x_n)\neq s'(x_n)$.
    Consequently, $\fd{c_1}{ x_n}$ entails that $s(c_1)\neq s'(c_1)$.
    These two assignments are illustrated in the top-left corner of~Figure~\ref{fig:c1}, where without loss of generality we use values $0$ and $1$. By local consistency, $\pX$ then contains two assignments $t$ and $t'$ in the context $x_1c_1x_2$ such that $t(x_1)=t'(x_1)$ and $t(c_1)\neq t'(c_1)$. Furthermore, $\fd{x_1}{x_2}$ entails $t(x_2)= t'(x_2)$. These assignments are illustrated in the top-middle table of Figure~\ref{fig:c1}, again using values $0$ and $1$. By local consistency and induction, the remaining contexts listed in \cref{item:c2} of the contextual chain rule (\ccr) contain assignment pairs of the form depicted in Figure~\ref{fig:c1}, with the values $b_i$ and $b'_i$ being unknown. In particular, \cref{item:fd1} of $\ccr$ guarantees that the two assignments on the remaining relations agree on the variables $x_i$, $i\in [n]$.

    Moving to the contexts listed in \cref{item:c3} of \ccr, the leftmost table in Figure~\ref{fig:c2} illustrates two assignments for the context $c_1c_2x_n$; the values of $c_1$ and $c_2$ are explained by local consistency with the top-right table of Figure~\ref{fig:c1}, and the values of $x_n$ are explained by $\fd{c_1}{ x_n}$ together with local consistency with the top-left table in Figure~\ref{fig:c1}.
    The assignments for the context $c_2c_3x_n$ in Figure~\ref{fig:c2} arise by local consistency with the assignments for the context $c_2x_2c_3$ in Figure~\ref{fig:c1},  %
    and from $\fd{c_2}{ x_n}$ together with local consistency with the leftmost table in Figure~\ref{fig:c2}. %
    Using this reasoning, by induction we eventually obtain the rightmost table in Figure~\ref{fig:c2}. 
    
    However, observe now that the pair of variables $(c_{n-1},x_n)$ is evaluated as $(0,0)$ and $(0,1)$ in \Cref{fig:c1,fig:c2}, respectively. By local consistency, this leads to a contradiction with the FD $\fd{c_{n-1}}{x_n}$. We thus conclude by contradiction that $\pX$  satisfies $\fd{x_1}{x_n}$.\end{proof}

\begin{figure}
	\centering
	\scalebox{0.8}{
		\setlength{\tabcolsep}{3pt} 
		\begin{tabular}{ccccccc}
			\begin{tabular}{ccc}
				$x_1$ & $c_1$ & $x_n$ \\
				\toprule
				0 & 0 & 0 \\
				0 & 1 & 1 \\
				\bottomrule
			\end{tabular}
			&
			\begin{tabular}{ccc}
				$x_1$ & $c_1$ & $x_2$ \\
				\toprule
				0 & 0 & 0 \\
				0 & 1 & 0 \\
				\bottomrule
			\end{tabular}
			&
			\begin{tabular}{ccc}
				$c_1$ & $x_2$ & $c_2$ \\
				\toprule
				0 & 0 & $b_2$ \\
				1 & 0 & $b'_2$ \\
				\bottomrule
			\end{tabular}
			\\[4ex]
			\begin{tabular}{ccc}
				$x_2$ & $c_2$ & $x_3$ \\
				\toprule
				0 & $b_2$ & 0 \\
				0 & $b'_2$ & 0 \\
				\bottomrule
			\end{tabular}
			&
			\begin{tabular}{ccc}
				$c_2$ & $x_2$ & $c_3$ \\
				\toprule
				$b_2$ & 0 & $b_3$ \\
				$b'_2$ & 0 & $b'_3$ \\
				\bottomrule
			\end{tabular}
			&
			$\dots$
			&
			\begin{tabular}{ccc}
				$x_{n-1}$ & $c_{n-1}$ & $x_n$ \\
				\toprule
				0 & $b_{n-1}$ & 0 \\
				0 & $b'_{n-1}$ & 0 \\
				\bottomrule
			\end{tabular}
		\end{tabular}
	}
	\caption{Relations arising from Items~\ref{item:fd1}, \ref{item:c1}, and \ref{item:c2}.}
	\label{fig:c1}
\end{figure}
    
\begin{figure}
	\centering
	\scalebox{0.8}{
		\setlength{\tabcolsep}{3pt}
		\begin{tabular}{cccc}
			\begin{tabular}{ccc}
				$c_1$ & $c_2$ & $x_n$ \\
				\toprule
				0 & $b_2$ & 0 \\
				1 & $b'_2$ & 1 \\
				\bottomrule
			\end{tabular}
			&
			\begin{tabular}{ccc}
				$c_2$ & $c_3$ & $x_n$ \\
				\toprule
				$b_2$ & $b_3$ & 0 \\
				$b'_2$ & $b'_3$ & 1 \\
				\bottomrule
			\end{tabular}
			&
			$\dots$
			&
			\begin{tabular}{ccc}
				$c_{n-2}$ & $c_{n-1}$ & $x_n$ \\
				\toprule
				$b_{n-2}$ & $b_{n-1}$ & 0 \\
				$b'_{n-2}$ & $b'_{n-1}$ & 1 \\
				\bottomrule
			\end{tabular}
		\end{tabular}
	}
	\caption{Relations arising from Items~\ref{item:fd2} and \ref{item:c3}.}
	\label{fig:c2}
\end{figure}

The proof of the proposition suggests that even soundness of inference rules can be non-trivial to check for general FDs.

\subsection{Complexity of unary FDs and arbitrary CDs}

In terms of complexity, one application of the contextual chain rule does not actually induce a higher complexity than the cycle rule from Section~\ref{subsec:Contextual_FDs}: we obtain $\mathrm{NL}$-completeness here as well.
For the containment, note that the conditions of the contextual chain rule can be checked iteratively, by guessing the necessary $x_i$ and $c_i$ values and only keeping constantly many of them in memory at a time.
For hardness, note that the first condition already requires an $x_1-x_n$ path in the underlying dependency graph.
If we ensure that there exists a $c$, such that all other conditions are trivially satisfied, we can simply reduce from $\mathrm{REACH}$ again, as we did in the case for the cycle rule.
The full proof can be found in the supplementary material.

\begin{proposition}
    \label{prop:1chain}
    The problem of deciding whether $\fd{x_1}{x_n}$ can be derived from a set of CDs and unary FDs $\Sigma$ by a single application of the contextual chain rule 
    is $\mathrm{NL}$-complete.
\end{proposition}

Even though deciding whether a single application of the chain rule yields a particular dependency can be checked in logarithmic space, the same idea does not yield a nondeterministic logspace algorithm for arbitrarily many applications of the rule.
The issue is that each application of the rule yields new FDs, which need to be kept in memory, and if a derivation requires polynomially many applications of the chain rule, polynomially many FDs need to be kept in memory.

We do---however---obtain a polynomial-time algorithm by simply checking for all pairs $x,y$, if $\fd{x}{y}$ is derivable from the current set of FDs and adding it, if it is.
This process is repeated until no new FDs can be added. 
For unary FDs this only takes polynomial time.

The full proof can be found in the supplementary material.

\begin{theorem}
    \label{thm:chain_ptime}
The problem whether a unary FD $\fd{y_1}{y_2}$ can be derived from a set $\Sigma$ of FDs and CDs by using the reflexivity rule, the cycle rule and the contextual chain rule is decidable in polynomial time.
\end{theorem}

\section{Conclusions and further directions}

 We considered realisability and logical inference in a theoretical framework for modelling local consistency in the presence of (potential) global inconsistency. The framework is based on locally consistent families of $K$-relations, for positive commutative monoids $K$. We investigated when a given family $\pX$ of locally consistent (set-based) relations can be realised as locally consistent $K$-families. In particular, we gave a complete characterisation for the case where $K$ is cancellative and $\pX$ is a contextual family over {hole context sequences}. Moreover, we established a sound and complete axiomatisation for the entailment of unary FDs in the setting of locally consistent $K$-relations. We discovered two weakenings of the transitivity rule---namely, the chain and cycle rules---that are sound in the contextual setting, while the usual transitivity rule of FDs becomes unsound. Finally, we established a polynomial time algorithm for checking whether a given unary FD is derivable from a given set of CDs and unary FDs by using the reflexivity rule, the cycle rule and the contextual chain rule. For unary FDs and binary CDs, we established that the entailment problem is NL-complete. The derivability problem of unary FDs and binary CDs remains NL-hard, even if only one application of either the contextual cycle or the contextual chain rule is allowed (see \Cref{table:results} on page \pageref{table:results} for an overview of our complexity and axiomatisation results).

\vspace{2mm}
\noindent\textbf{Beyond unary FDs.}
Until now, we have focused on the entailment problem for unary FDs. While these are the simplest kind of FDs, we have already revealed that even this case is far from trivial, and includes context specific intricacies.

We formulated our novel axiom schemas---the contextual cycle rule and the contextual chain rule---to be applied to single variables only. However, the rules remain sound even if the scopes of the rules are extended to sets of variables.
\begin{remark}\label{rem:generalcycle}
The cycle rule extends from single variables to sets of variables.
That is, for contextual $K$-families, if
\[
x_1 \to x_2 \to \dots \to x_k \to x_1,
\]
then $x_1 \to x_k$, for all \emph{sets} of variables $x_1,\dots,x_k$ with $k \ge 1$.
Likewise, the contextual chain rule extends from variables to sets of variables.
Specifically, the proofs of Lemma~\ref{lem:sound} and Proposition~\ref{prop:advanced_chain_rule} do not rely on the variables being
singletons, and apply verbatim to sets of variables. In particular, for a set of variables $X$, we simply interpret $t(X)=t'(X)$ as $\forall x\in X\colon t(x)=t'(x)$.
\end{remark}
With the above extension, we can see that
the contextual transitivity rule is an instance of the contextual chain rule by setting $n=3$, $x_1=x$, $x_2=y$, $x_3=z$, $c_1=y=c_2$ (as we already mentioned in Section \ref{sec:path}).
Hence reflexivity, augmentation, and the contextual chain rule constitute a sound axiomatisation in the contextual setting, that is complete for derivations within a single context (cf. \Cref{prop:CAtoA}\footnote{The proof can be adapted to the variant of CA, where the contextual transitivity rule is replaced with the contextual chain rule. The left-to-right direction will utilise the fact that Armstrong's axioms are complete within a single context, while the right-to-left direction can be applied verbatim, since the contextual transitivity rule is an instance of the contextual chain rule.}).

\vspace{2mm}
\noindent\textbf{Future work.}
We conclude by listing some interesting questions deserving further study:

\begin{itemize}
\item Does there exist a characterisation akin to Theorem \ref{thm:nosup} for arbitrary context sets?
\item Can a more general decomposition theorem---analogous to Proposition \ref{prop:combination}---be obtained for $\Rp$-families?
\item What is the structure of the equivalence relation on monoids that is defined by $K$-realisability of contextual families (cf. Prop. \ref{prop:char})?
E.g., if the monoid has a non-zero idempotent element $a$ (i.e., $a+a=a$), then every contextual family is realisable. What can be said about monoids without idempotent elements, apart from $0$?
\item The chain rule and the cycle rule are presented as infinite axiom schemas.
Can these rules be simplified? Are they implied by some more fundamental rule? Can we replace the infinite schemas with a finite set of rules?
\item Is the full (non-unary) entailment problem for FDs invariant of the monoid (cf. Corollaries \ref{cor:simple} and \ref{cor:RtoN})? For any monoid with a non-zero idempotent element, the entailment problem coincides with that of the Boolean monoid.
\item Is the entailment problem finitely axiomatisable and decidable? For the Boolean monoid,  this is a subproblem of the undecidable implication problem for FDs and inclusion dependencies~\cite{chandra85}. Note that, inclusion dependencies (see e.g., \cite{abi95}) can be used to express the local consistency property (in the non-contextual setting). Hence the intricacies that arise in our axioms for FDs arise already in the non-contextual setting, if FDs are axiomatised together with inclusion dependencies.  
\end{itemize}
Recent work by Hirvonen \shortcite{DBLP:journals/tocl/Hirvonen24,DBLP:journals/corr/abs-2510-08112} considers the implication problem and the axiomatisability of functional dependencies together with unary inclusion dependencies, also in the monoidal setting. As this analysis is restricted %
to the uni-relational setting, their results do not translate to our more general framework.
Nevertheless, generalising their work to our setting may prove fruitful.

\section*{Acknowledgments}
Miika Hannula was partially supported by the ERC grant 101020762. Juha Kontinen was partially funded by grant 359650 of The Research Council of Finland.
Timon Barlag, Nina Pardal, and Jonni Virtema were partially supported by the DFG grant VI 1045-1/1.
Timon Barlag was partially supported by the DAAD project 57710940.

\bibliographystyle{kr}
\bibliography{biblio}

\begin{thebibliography}{}

\bibitem[\protect\citeauthoryear{Abiteboul, Hull, and Vianu}{1995}]{abi95}
Abiteboul, S.; Hull, R.; and Vianu, V.
\newblock 1995.
\newblock {\em Foundations of Databases}.
\newblock Addison-Wesley.

\bibitem[\protect\citeauthoryear{Abramsky and
  Brandenburger}{2011}]{Abramsky2011-ABRTSS}
Abramsky, S., and Brandenburger, A.
\newblock 2011.
\newblock The sheaf-theoretic structure of non-locality and contextuality.
\newblock {\em New Journal of Physics} 13:113036.

\bibitem[\protect\citeauthoryear{Abramsky and Hardy}{2012}]{PhysRevA.85.062114}
Abramsky, S., and Hardy, L.
\newblock 2012.
\newblock Logical bell inequalities.
\newblock {\em Phys. Rev. A} 85.

\bibitem[\protect\citeauthoryear{Abramsky \bgroup et al\mbox.\egroup
  }{2015}]{DBLP:conf/csl/AbramskyBKLM15}
Abramsky, S.; Barbosa, R.~S.; Kishida, K.; Lal, R.; and Mansfield, S.
\newblock 2015.
\newblock Contextuality, cohomology and paradox.
\newblock In Kreutzer, S., ed., {\em 24th {EACSL} Annual Conference on Computer
  Science Logic, {CSL} 2015, September 7-10, 2015, Berlin, Germany}, volume~41
  of {\em LIPIcs},  211--228.
\newblock Schloss Dagstuhl - Leibniz-Zentrum f{\"{u}}r Informatik.

\bibitem[\protect\citeauthoryear{Abramsky}{2013}]{Abramsky13}
Abramsky, S.
\newblock 2013.
\newblock Relational hidden variables and non-locality.
\newblock {\em Stud Logica} 101(2):411--452.

\bibitem[\protect\citeauthoryear{Armstrong}{1974}]{Armstrong_Axioms}
Armstrong, W.~W.
\newblock 1974.
\newblock Dependency structures of data base relationships.
\newblock In {\em IFIP Congress},  580--583.

\bibitem[\protect\citeauthoryear{Atserias and Kolaitis}{2021}]{AK21}
Atserias, A., and Kolaitis, P.~G.
\newblock 2021.
\newblock Structure and complexity of bag consistency.
\newblock In {\em Proceedings of the 40th ACM SIGMOD-SIGACT-SIGAI Symposium on
  Principles of Database Systems}, PODS'21,  247–259.
\newblock New York, NY, USA: Association for Computing Machinery.

\bibitem[\protect\citeauthoryear{Atserias and Kolaitis}{2025}]{10.1145/3651608}
Atserias, A., and Kolaitis, P.~G.
\newblock 2025.
\newblock Consistency of relations over monoids.
\newblock {\em J. {ACM}} 72(3):18:1--18:47.

\bibitem[\protect\citeauthoryear{Beeri and Bernstein}{1979}]{BeeriB79}
Beeri, C., and Bernstein, P.~A.
\newblock 1979.
\newblock Computational problems related to the design of normal form
  relational schemas.
\newblock {\em {ACM} Trans. Database Syst.} 4(1):30--59.

\bibitem[\protect\citeauthoryear{Beeri \bgroup et al\mbox.\egroup
  }{1983}]{10.1145/2402.322389}
Beeri, C.; Fagin, R.; Maier, D.; and Yannakakis, M.
\newblock 1983.
\newblock On the desirability of acyclic database schemes.
\newblock {\em J. ACM} 30(3):479–513.

\bibitem[\protect\citeauthoryear{Bell}{1964}]{Bell:1964kc}
Bell, J.~S.
\newblock 1964.
\newblock {On the Einstein-Podolsky-Rosen paradox}.
\newblock {\em Physics Physique Fizika} 1:195--200.

\bibitem[\protect\citeauthoryear{Chandra and Vardi}{1985}]{chandra85}
Chandra, A.~K., and Vardi, M.~Y.
\newblock 1985.
\newblock The implication problem for functional and inclusion dependencies is
  undecidable.
\newblock {\em SIAM Journal on Computing} 14(3):671--677.

\bibitem[\protect\citeauthoryear{Console and Lenzerini}{2014}]{ecai-ConsoleL14}
Console, M., and Lenzerini, M.
\newblock 2014.
\newblock Reducing global consistency to local consistency in ontology-based
  data access.
\newblock In Schaub, T.; Friedrich, G.; and O'Sullivan, B., eds., {\em {ECAI}
  2014 - 21st European Conference on Artificial Intelligence, 18-22 August
  2014, Prague, Czech Republic - Including Prestigious Applications of
  Intelligent Systems {(PAIS} 2014)}, volume 263 of {\em Frontiers in
  Artificial Intelligence and Applications},  219--224.
\newblock {IOS} Press.

\bibitem[\protect\citeauthoryear{Hirvonen}{2024}]{DBLP:journals/tocl/Hirvonen24}
Hirvonen, M.
\newblock 2024.
\newblock The implication problem for functional dependencies and variants of
  marginal distribution equivalences.
\newblock {\em {ACM} Trans. Comput. Log.} 25(4):1--23.

\bibitem[\protect\citeauthoryear{Hirvonen}{2025}]{DBLP:journals/corr/abs-2510-08112}
Hirvonen, M.
\newblock 2025.
\newblock Implication problems over positive semirings.
\newblock {\em CoRR} abs/2510.08112.

\bibitem[\protect\citeauthoryear{Howard \bgroup et al\mbox.\egroup
  }{2014}]{Howard2014}
Howard, M.; Wallman, J.; Veitch, V.; and Emerson, J.
\newblock 2014.
\newblock Contextuality supplies the ‘magic’ for quantum computation.
\newblock {\em Nature} 510(7505):351--355.

\bibitem[\protect\citeauthoryear{Hyttinen, Paolini, and
  Väänänen}{2015}]{HYTTINEN}
Hyttinen, T.; Paolini, G.; and Väänänen, J.
\newblock 2015.
\newblock Quantum team logic and {B}ell’s inequalities.
\newblock {\em The Review of Symbolic Logic} 8(4):722–742.

\bibitem[\protect\citeauthoryear{Kishida}{2016}]{kishida:LIPIcs.ICALP.2016.113}
Kishida, K.
\newblock 2016.
\newblock Logic of local inference for contextuality in quantum physics and
  beyond.
\newblock In {\em {ICALP}}, volume~55 of {\em LIPIcs},  113:1--113:14.
\newblock Schloss Dagstuhl - Leibniz-Zentrum f{\"{u}}r Informatik.

\bibitem[\protect\citeauthoryear{Kochen and Specker}{1967}]{06479255de6a}
Kochen, S., and Specker, E.~P.
\newblock 1967.
\newblock The problem of hidden variables in quantum mechanics.
\newblock {\em Journal of Mathematics and Mechanics} 17(1):59--87.

\bibitem[\protect\citeauthoryear{Koller and Friedman}{2009}]{10.5555/1795555}
Koller, D., and Friedman, N.
\newblock 2009.
\newblock {\em Probabilistic Graphical Models: Principles and Techniques -
  Adaptive Computation and Machine Learning}.
\newblock The MIT Press.

\bibitem[\protect\citeauthoryear{Lembo \bgroup et al\mbox.\egroup
  }{2015}]{ISWC-LemboMRST15}
Lembo, D.; Mora, J.; Rosati, R.; Savo, D.~F.; and Thorstensen, E.
\newblock 2015.
\newblock Mapping analysis in ontology-based data access: Algorithms and
  complexity.
\newblock In Arenas, M.; Corcho, {\'{O}}.; Simperl, E.; Strohmaier, M.;
  d'Aquin, M.; Srinivas, K.; Groth, P.; Dumontier, M.; Heflin, J.;
  Thirunarayan, K.; and Staab, S., eds., {\em The Semantic Web - {ISWC} 2015 -
  14th International Semantic Web Conference, Bethlehem, PA, USA, October
  11-15, 2015, Proceedings, Part {I}}, volume 9366 of {\em Lecture Notes in
  Computer Science},  217--234.
\newblock Springer.

\bibitem[\protect\citeauthoryear{Papadimitriou}{1994}]{DBLP:books/daglib/0072413}
Papadimitriou, C.~H.
\newblock 1994.
\newblock {\em Computational complexity}.
\newblock Addison-Wesley.

\bibitem[\protect\citeauthoryear{Paw{\l}owski \bgroup et al\mbox.\egroup
  }{2009}]{pawlowski2009information}
Paw{\l}owski, M.; Paterek, T.; Kaszlikowski, D.; Scarani, V.; Winter, A.; and
  {\.Z}ukowski, M.
\newblock 2009.
\newblock Information causality as a physical principle.
\newblock {\em Nature} 461(7267):1101--1104.

\bibitem[\protect\citeauthoryear{Schrijver}{1999}]{0090562}
Schrijver, A.
\newblock 1999.
\newblock {\em Theory of linear and integer programming}.
\newblock Wiley-Interscience series in discrete mathematics and optimization.
  Wiley.

\bibitem[\protect\citeauthoryear{Vorob’ev}{1962}]{V62}
Vorob’ev, N.~N.
\newblock 1962.
\newblock Consistent families of measures and their extensions.
\newblock {\em Theory of Probability \& Its Applications} 7(2):147--163.

\end{thebibliography}

\iffull
\else
\cleardoublepage
\appendix

\fi

\section{Full proofs}

Full proof of Theorem \ref{thm:isolated}.
\isolated*
\begin{proof}
The direction from right-to-left is immediate. We will prove the converse.
Assume $\Sigma \vdash_\textrm{AX} \fd{x}{y}$. Therefore there is a sequence of FDs $\psi_0,\dots, \psi_{n-1}$ such that
\begin{enumerate}
    \item $\psi_{n-1}=\fd{x}{y}$
    \item for $i\in [n]$, $\psi_i$ is from $\Sigma$ or obtained from $\{\psi_0, \dots ,\psi_{i-1}\}$ by using one of the rules of $\mathrm{AX}$.
    \item for $i\in [n]$, $\Vars(\psi_i)\in \Cont(\Sigma \cup\{\fd{x}{y}\})$.
\end{enumerate}
For $i\in [n]$, define $\Sigma_{i}:= \{\psi_0,\dots,\psi_{i-1}\}$. 
If we can show that $\Sigma_{i}\vdash_{\mathrm{CA}}\psi_i$, for every $i\in n$ such that $\psi_i\not\in\Sigma$, we can combine these derivations to obtain a proof of $\Sigma\vdash_{\mathrm{CA}}\fd{x}{y}$.

To that end, let $i\in [n]$ be such that  $\psi_i\not\in\Sigma$. Thus, $\psi_i$ is derivable from $\Sigma_{i}$ using some uni-context rule $(*)\in \mathrm{AX}$. Therefore, there exists a subset $\Sigma'_{i} \subseteq \Sigma_{i}$ such that $\Sigma'_{i}\cup\{\psi_i\}$ is uni-context and 
$\Sigma'_{i}\vdash_{\mathrm{AX}}\psi_i$. Set $z:= \Vars(\Sigma'_{i}\cup\{\psi_i\})$ and note that $z\in \Cont(\Sigma'_{i}\cup\{\psi_i\})$.
Since (non-contextual) relations are special cases of contextual families and since axioms in $\mathrm{AX}$ are sound, it follows that $\Sigma'_{i}\models\psi_i$ holds in the relational (non-contextual) setting.
Since Armstrong's axioms are complete, it follows that $\Sigma'_{i}\vdash_{\mathrm{A}}\psi_i$. Now, by \Cref{prop:CAtoA}, we obtain $\Sigma'_{i}\cup\{\fd{z}{z}\}\vdash_{\mathrm{CA}}\psi_i$ and by the reflexivity axiom, that $\Sigma'_{i} \vdash_{\mathrm{CA}} \{\fd{z}{z}\}$. By \Cref{rem:composition}, we obtain $\Sigma'_{i} \vdash_{\mathrm{CA}} \psi_i$, which implies $\Sigma_{i} \vdash_{\mathrm{CA}} \psi_i$.
\end{proof}

Full proof of Theorem \ref{thm:completeness}.
\completeness*
\begin{proof}
We prove the case $K=\B$; the proof of the general case is identical, except that %
all relation constructions have to be replaced with uniformly annotated $K$-relations, that is every assignment is annotated with the same monoid value in $K$. Also, note that $\pX\models \fd{x}{y}$ iff $\supp(\pX)\models \fd{x}{y}$ holds by definition.

First note that, if $\phi$ is a CD, then both $\Sigma\models \phi$ and $\Sigma \vdash_{\mathcal{CR}} \phi$ trivially hold. Thus we assume that $\phi$ is a unary FD.
Let us consider soundness first. Suppose $(\psi_1, \dots ,\psi_n)$ is a proof of $\psi_n$ from $\Sigma$, and let $\pX$ be a contextual family over $\cons$ such that $\pX \models \Sigma$, and $\Vars(\theta)\in \cons$ for all $\theta \in \Sigma \cup \{\psi_n\}$. We need to show that $\pX \models \psi_n$.
As the induction hypothesis, suppose $\pX \models \psi_i$ for all $i< j$, where $j\leq n$. Suppose also $\psi_j$ is obtained from $\psi_{\ell_1}, \dots, \psi_{\ell_n}$, where $\ell_1 < \dots <\ell_n < j$, by using the cycle rule. In particular, the induction hypothesis entails that $\pX$ satisfies all of $\psi_{\ell_1}, \dots, \psi_{\ell_n}$ and their contexts are in $\cons$ (due to \Cref{it:new} of \Cref{def:derivation}). Hence $\Vars(\psi_j)\in \cons$, due to the form of the cycle rule. 
From \Cref{lem:sound} it now follows that $\pX \models \psi_j$, for $\pX$ contains all the necessary contexts.
The case for reflexivity is immediate by  \Cref{it:new} of \Cref{def:derivation}, 
concluding the induction proof.

    For completeness, we prove the contraposition: $\Sigma \not\models \fd{x}{y}$ if $\Sigma \not\vdash \fd{x}{y}$. Assume $\Sigma \not\vdash \fd{x}{y}$. By reflexivity, $x$ and $y$ are different variables. Let $V$ be the set of variables that appear in $\Sigma \cup \{x,y\}$.
    Consider the directed graph $G$ with node set $V$ and edges consisting of the functional dependencies $\fd{u}{v}$ appearing in $\Sigma$.
    There are two cases:
    
    \textbf{Case 1.} There is no path from $x$ to $y$ in $G$. Let $X$ be the set of variables to which there is a path from $x$, including $x$ itself.
    We create one relation $R$ satisfying $\Sigma$ and falsifying $\fd{x}{y}$. This relation has two assignments: $s$ that maps each variable from $V$ to $0$, and $s'$ that maps each variable from $X$ to $0$ and the remaining variables from $V \setminus X$ to $1$. Since $x \in X$ and $y \in V \setminus X$, we obtain $R \not\models \fd{x}{y}$.
    For every $\fd{u}{v}\in \Sigma$, we have $R \models \fd{u}{v}$, otherwise $u \in X$ and $v \in V \setminus X$, contradicting the definition of $X$. By taking the projections of $R$ on all subsets of $V$, we obtain a contextual family $\pX$ witnessing $\Sigma \not\models \fd{x}{y}$.

    \textbf{Case 2.} There is a path from $x$ to $y$ in $G$. We construct a counterexample contextual family  $\pX$ as follows.    
    For each $\fd{u}{v}$ or $\fd{uv}{uv}$ appearing in $\Sigma$ such that $\{u,v\}\neq \{x,y\}$, construct a relation that has two assignments mapping $(u,v)$ respectively to $(0,0)$ and $(1,1)$. For $\fd{x}{y}$, construct a relation that has four assignments mapping $(x,y)$ respectively to $(0,0),(0,1),(1,0),(1,1)$. The local consistency criterion is satisfied, since each variable has values $0$ and $1$ in every relation in which it appears.
    Let $\pX$ denote the contextual family that arises from the aforementioned relations. 
    Clearly $\pX\not\models\fd{x}{y}$, but $\pX\models\fd{u}{v}$, for each $\fd{u}{v}\in \Sigma$ such that $\{u,v\}\neq \{x,y\}$. It remains to consider the case of $\fd{u}{v}\in \Sigma$, where $\{u,v\}= \{x,y\}$. 
    Clearly, it cannot be that $\fd{u}{v}= \fd{x}{y}$, since this would contradict the assumption that $\Sigma \not\vdash \fd{x}{y}$. Moreover, if $\fd{u}{v}= \fd{y}{x}$, then the path from $x$ to $y$ and the cycle rule entails $\Sigma \vdash \fd{x}{y}$, contradicting our assumption.
    Hence we conclude that all FDs in $\Sigma$ are satisfied. The resulting contextual family $\pX$ thus provides a proof for $\Sigma \not\models \fd{x}{y}$, concluding the second case and the proof. In the general case of $K\neq \B$, using positivity and non-triviality of $K$ we can select two non-zero monoid values $a,b\in K$ such that $a+a=b$. 
    Then assignments in the above defined ($K$-)relations of cardinality $2$ shall be annotated with $b$, while assignments in the ($K$-)relation of cardinality $4$ shall be annotated with $a$. It is easy to see that with this amendment the proof goes through also in the general case.
     \end{proof}

   For the proof of Theorem \ref{thm:nosup},  we need some auxiliary concepts and results.

 A digraph \( G =(V,E)\) is called \emph{cyclically \( n \)-partite} if there exists a partition
\(
\{ V_0, \dots, V_{n-1} \}
\)
of $V$ such that for every edge $(u,v)\in E$ we have
\(
u \in V_i\) and  $v \in V_{i+1}$, for some $i \in [n-1]$.
The following lemma is a simple observation. 
\begin{lemma}\label{lem:easy}
Let $\pX$ be a contextual family over a hole $\cons=(C_0, \dots ,C_{n-1})$. 
Then $\opgraph{\pX,\cons}$ is cyclically \( n \)-partite.
\end{lemma}
\begin{proof}
The claim follows from the definition of $\opgraph{\pX,\cons}$, and the required partition is %
\begin{equation*}%
\{\pX_1 \restr (C_1\cap C_{2}), \dots, \pX_{n-1} \restr (C_{n-1}\cap C_{n}),\pX_n \restr (C_n\cap C_{1})\}.
\end{equation*}
\end{proof}

Then, we show that every cycle gives rise to a uniformly weighted $K$-relation.
Recall that, by definition, a (simple) cycle $(v_1, \dots ,v_\ell,v_1)$ %
does not visit any vertex twice, except for the first and the last one.
We say that a contextual family $\pX$ over a hole is \emph{cyclic} if $\opgraph{\pX,\cons}$ contains a cycle 
\begin{equation}\label{eq:simple}
v_1 \ledge{s_1} v_2 \ledge{s_2} \dots \ledge{s_{\ell-2}} v_{\ell-1}\ledge{s_{\ell-1}}v_\ell \ledge{s_\ell} v_1,
\end{equation}
where $\set{\pX}=\{s_1, \dots ,s_\ell\}$.
The next lemma states that any cyclic family can be realised as a $K$-family using uniform multiplicities.
\begin{lemma}\label{lem:help}   
 Let $\pX$ be a cyclic contextual family over a hole  $\cons=(C_0, \dots ,C_{n-1})$. Let $K$ be a positive commutative monoid and $0 \neq w\in K$.
    Then $w\pX$ is a contextual $K$-family, and $\supp(w\pX)=\pX$.
    \end{lemma}

\begin{proof}
    We need to prove that $w\pX$ is locally consistent. 
    Suppose the cycle has length $\ell$, as in \eqref{eq:simple}. Since $\opgraph{\pX,\cons}$ is cyclically $n$-partite (\Cref{lem:easy}), $\ell$ must be a multiple of $n$; that is, $n\cdot m=\ell$ for some positive integer $m$.
    Fix $i \in [n]$.
    One observes that 
    \[
   |\pX_i\restr (C_{i-1} \cap C_{i})|=|\pX_i| =|\pX_i\restr (C_i \cap C_{i+1})|=m.
    \]
    Moreover, each element from $\pX_i\restr (C_i \cap C_{i+1})$ appears in $\pX_{i+1}\restr (C_i \cap C_{i+1})$, and conversely each element of $\pX_{i+1}\restr (C_i \cap C_{i+1})$ appears in $\pX_i\restr (C_i \cap C_{i+1})$. We conclude that $(w\pX)_i
 \restr (C_i \cap C_{i+1}) = (w\pX)_{i+1} \restr (C_i \cap C_{i+1})$. If $i$ and $j$ are not adjacent, then $C_i \cap C_j= \emptyset$. For this case it suffices to note that $(w\pX)_i\restr \emptyset$ consists of the empty assignment associated with the sum ${w+\dots +w}$ of length $m$. Thus $w\pX$ is a contextual $K$-family. It follows by positivity of $K$ that $\supp(w\pX)=\pX$.
\end{proof}

\nosup*
\begin{proof}
( $\Rightarrow$ )
    Let $\pX$ be the support of some contextual $K$-family $\pY$. Then $\pY$ can be represented as $\#\colon \set{\pX}\to K_{\neq 0}$, a mapping which we extend to sets $S$ by $\#(S)\coloneqq \sum_{s\in S}\#(s)$. Toward a contradiction, assume $\opgraph{\pX,\cons}$ is not edge-cyclic.  Then, some $s \in \set{\pX}$ gives rise to an edge that is not in any cycle. 
    Without loss of generality, we assume that $s$ is associated with the context $C_0$, that is, it belongs to $\pX_0$.
    We define sets of assignments $L^j$, $j\geq 0$, inductively as follows:
    \begin{itemize}
        \item $L^0\coloneqq \{s' \in \pX_{{1}} \mid {s} \restr {(C_0 \cap C_{1})}={s'}\restr{(C_0 \cap C_{1})}\}$,
        \item $L^{j+1}$ extends $L^j$ with assignments $s'\in \pX_{{j+1}}$, $j\in [n]$, such that there is $s \in L^i \cap \pX_{{j}}$ with
        \[ {s} \restr {(C_j \cap C_{j+1})}={s'}\restr{(C_j \cap C_{j+1})}.\]
    \end{itemize}
    Now, let $M \coloneqq \bigcup_{j=1}^\infty L^j$. Since $\pX$ is finite, there exists a positive integer $k$ such that $M= L^k$. 
    Furthermore, we write $M_i \coloneqq \pX_{i}\cap M$, for $i\in [n]$. That is, $M_i$ consists of the assignments of $M$ over the context $C_i$.

    We claim that the following inequalities hold:
    \begin{equation}\label{eq:seq}
    \hspace{-2mm}
    \# (M_0 \cup \{s\}) \leq \#(M_1) \leq \ldots  \leq \# (M_{n-1}) \leq \#(M_0),
    \end{equation}
    where $\leq$ refers to the natural (pre)order of $K$. Note that if this claim holds, we reach a contradiction. 
    Indeed, note that $s \notin L_0$ by hypothesis, hence the transitivity of $\leq$ implies $\# (M_0)+ \#(\{s\})=\# (M_0 \cup \{s\})\leq \#(M_0)$. 
    Thus, there is $c \in K$ such that $\# (M_0)+ \#(\{s\})+ c= \#(M_0)$. By cancellativity, $\#(\{s\})+ c=0$, and hence by positivity, $\#(s)=0$. This contradicts the fact that  $s \in \set{\pX}$, where $\#$ is a function from $\set{\pX}$ to $ K_{\neq 0}$.

    Thus it remains to prove \eqref{eq:seq}. First, define
    \begin{itemize}
        \item $\pX_{i,+} =\{s' \in \pX_i\mid  s' \restr {(C_i \cap C_{i+1})}\in M_i\restr (C_i \cap C_{i+1})\}$,
        \item $
       \pX_{i+1,-} =\{s' \in \pX_{i+1}\mid s' \in {(C_i \cap C_{i+1})}\in M_i\restr (C_i \cap C_{i+1})\}.$
    \end{itemize}
    Then we observe that, for $i\in \{1, \dots ,n-1\}$,
     \begin{multicols}{3}
    \begin{enumerate}
    \item\label{it:eka} $M_i \subseteq \pX_{i,+}$,
\item\label{it:toka} $\#(\pX_{i,+})  = \# (\pX_{i+1,-})$, %
\item\label{it:kolmas} $\pX_{i+1,-}= M_{i+1}$.
    \end{enumerate}
        \end{multicols}
 \Cref{it:eka} is immediate, \cref{it:toka} follows by local consistency, and \cref{it:kolmas} follows by the  construction of $M$. 
 Since $\leq$ is transitive and reflexive, we obtain $\#(M_i)\leq \#(M_{i+1})$.
    Also, for the first inequality in \eqref{eq:seq}, we obtain $M_0\cup \{s\} \subseteq \pX'_{0,+}$, $\#(\pX'_{0,+})  = \# (\pX'_{1,-})$, $\pX'_{1,-}= M_{1}$, where 
 $\pX'_{0,+},\pX'_{1,-}$ are defined otherwise as $\pX_{0,+},\pX_{1,-}$, except that $M_0$ is replaced with $M_0 \cup \{s\}$. This concludes the proof of the ``$\Rightarrow$'' direction

 ( $\Leftarrow$ ) 
 Suppose $\opgraph{\pX,\cons}$ is edge-cyclic. %
 Then, each assignment $s \in \set{\supp(\pX)}$ belongs to a cycle in $\opgraph{\pX,\cons}$.
 In particular, there must be a cyclic contextual family $\pZ_s$ such that $s\in \set{\supp(\pZ_s)}\subseteq \set{\supp(\pX)}$. %
Then, \Cref{lem:help} entails $w\pZ_s$ is a contextual $K$-family, for any $w \neq 0$ (which exists since $K$ is non-trivial).
Since contextual $K$-families over $\cons$ are closed under addition (Prop.~\ref{prop:closedplus}),
the sum $\sum_{s\in \set{\supp(\pX)}} w\pZ_s$ is a contextual $K$-family. Its support is $\pX$, for $K$ is positive. Hence the ``$\Leftarrow$'' follows.
\end{proof}

Full proof for Proposition~\ref{prop:1chain}.

\begin{manualprop}{\ref{prop:1chain}}
    The problem of deciding whether $\fd{x_1}{x_n}$ can be derived from a set of CDs and unary FDs $\Sigma$ by a single application of the contextual chain rule 
    is $\mathrm{NL}$-complete.
\end{manualprop}

\begin{proof}
    To verify that the problem is in $\mathrm{NL}$, we need to show that it can be decided in nondeterministic logarithmic space.
    To that end, it is crucial that each of the five conditions of the contextual chain rule can be checked step by step, without the need of keeping the previous steps in memory.
    Given a set of CDs and unary FDs $\Sigma$ and two variables $x_1$ and $x_n$, the algorithm to solve the aforementioned decision problem works as follows:

    First, we nondeterministically guess $c_1$ and check whether $\fd{c_1}{x_n} \in \Sigma$.
    Afterwards, for all $i \in (2, \dots, \mbox{n-1})$, we guess $c_i$ and $x_i$ and check if $\fd{x_{i-1}}{x_i}, \fd{c_i}{x_n}, x_{i-1}c_{i-1}x_i, c_{i-1}x_ic_i, c_{i-1}c_i,x_n \in \Sigma$.
    Finally, we check whether $x_{n-1}c_{n-1}x_n \in \Sigma$.

    If all those checks are passed, we know that there exist $x_2, \dots, x_{-1}$ and $c_1, \dots, c_{n-1}$ such that the conditions of the rule are met and thus $\fd{x_1}{x_n}$ can be derived from $\Sigma$ by a single application of the contextual chain rule.
    Furthermore, we only ever had to keep four non-input values, i.e., $x_{i-1}, x_i, c_{i-1}$ and $c_i$, in memory, for which logarithmic space is sufficient.

    To show $\mathrm{NL}$-hardness, we reduce from the well-known $\mathrm{NL}$-complete graph reachability problem~\cite{DBLP:books/daglib/0072413}: for a graph $G$ and two nodes $s, t$, $(G, s, t) \in \mathrm{REACH}$ if and only if there is a path from $s$ to $t$ in $G$.

    Note that the first condition of the contextual chain rule essentially encapsulates reachability from $x_1$ to $x_n$ in the dependency graph of $\Sigma$.
    The four further conditions only add the requirement that all those ``edges'' need to be ``witnessed'' by a variable $c_i$.
    If we transform a graph into a set of dependence atoms where each ``edge'' is trivially ``witnessed'', we end up with a plain reachability problem again.

    Given $G=(V,E)$ and two nodes $s$ and $t$, we produce the unary FDs $\{\fd{u}{v} \mid (u,v) \in E\} \cup \{c \to t\}$ and the CDs $\{uvc \mid (u, v) \in E\}$.
    In the resulting set $\Sigma$, there is a chain of FDs from $s$ to $t$, if and only if there is a  path from $s$ to $t$ in $G$. 
    And since all edges are witnessed, the four conditions beyond the first are trivially met.
    Finally, this function is a logspace reduction, since only one edge at a time needs to be kept in memory.
    In total, this yields $\mathrm{NL}$-completeness of one application of the chain rule.
\end{proof}

Full proof for Theorem~\ref{thm:chain_ptime}.

\begin{manualtheorem}{\ref{thm:chain_ptime}}
The problem whether a unary FD $\fd{y_1}{y_2}$ can be derived from a set $\Sigma$ of FDs and CDs by using the reflexivity rule, the cycle rule and the contextual chain rule is decidable in polynomial time.
\end{manualtheorem}

\begin{proof}

    We compute iteratively the set $D$ of FDs and CDs that can be derived from $\Sigma$ and could then be utilised in the possible proof of $\Sigma\vdash \fd{y_1}{y_2}$ in the sense of Definition~\ref{def:derivation}.
    Initially, $D =\Sigma$.

    First, we compute all CDs of size at most $3$ that derivable from $\Sigma$. 
    That is, we add a CD of the form $x_1x_2x_3$ to $D$, if there exists an FD or a CD in $\Sigma\cup\{\fd{y_1}{y_2}\}$ that contains at least the variables $x_1$, $x_2$ and $x_3$.
    They are derivable by using the reflexivity rule and there are only polynomially many such CDs to check.
    We do this, so that the contextual chain rule can be applied.

    We then proceed as follows:
    For all pairs of variables $x, y$, such that $\fd{x}{y} \notin D$, check whether $D$ entails $\fd{x}{y}$ by the contextual chain rule or the cycle rule.
    If so, add $\fd{x}{y}$ to $D$.
    Repeat this process, until $D$ remains unchanged for all remaining variable pairs.
    Finally, check whether $\fd{y_1}{y_2}\in D$, which decides whether $\Sigma\vdash \fd{y_1}{y_2}$.

    There are only polynomially many pairs of variables and each rule takes polynomially many steps to verify, so each pass takes polynomial time.
    Since there are only polynomially many variables, this ``derivation closure'' can be at most polynomially large, implying that it can only change in polynomially many passes in this process.
    Therefore, the entire process takes polynomial time.
\end{proof}

\end{document}